\documentclass[11pt]{article}

\usepackage{booktabs} 

\usepackage[ruled]{algorithm2e} 

\SetAlFnt{\small}
\SetAlCapFnt{\small}
\SetAlCapNameFnt{\small}
\SetAlCapHSkip{0pt}
\IncMargin{-\parindent}

\usepackage{amsthm,amsmath,amssymb}
\usepackage{xcolor}
\usepackage{graphicx}
\usepackage{booktabs}
\usepackage{fullpage}
\usepackage{caption}
\usepackage{hyperref}
\usepackage{enumerate}
\usepackage{enumitem}

\newcommand{\eps}{\varepsilon}

\newcommand{\skel}{\text{skeleton}}
\newcommand{\core}{\text{core}}
\newcommand{\cluster}[1]{\langle #1 \rangle}
\renewcommand{\phi}{\varphi}

\newtheorem{lemma}{Lemma}

\newtheorem{theorem}{Theorem}

\newtheorem{corollary}{Corollary}

\begin{document}
	
\title{Recognizing Weak Embeddings of Graphs\thanks{Research supported in part by the NSF awards CCF-1422311 and CCF-1423615, and the Science Without Borders program.  The second author gratefully acknowledges  support from Austrian Science Fund (FWF): M2281-N35. A preliminary version of this paper appeared in the \emph{Proceedings of the 29th ACM-SIAM Symposium on Discrete Algorithms} (SODA), 2018, pp.~274--292, \url{https://doi.org/10.1137/1.9781611975031.20}.}}

\author{
	Hugo A. Akitaya\thanks{Department of Computer Science, Tufts University, Medford, MA, USA.}
	\and
	Radoslav Fulek\thanks{Institute of Science and Technology, Klosterneuburg, Austria.}
	\and
	Csaba D. T\'oth\footnotemark[2]~\thanks{Department of Mathematics, California State University Northridge, Los Angeles, CA, USA.}
}

\maketitle

	\begin{abstract}
		We present an efficient algorithm for a problem in the interface between clustering and graph embeddings. An \textbf{embedding} $\varphi:G\rightarrow M$ of a graph $G$ into a 2-manifold $M$ maps the vertices in $V(G)$ to distinct points and the edges in $E(G)$ to interior-disjoint Jordan arcs between the corresponding vertices. In applications in clustering, cartography, and visualization, nearby vertices and edges are often bundled to the same point or overlapping arcs, due to data compression or low resolution. This raises the computational problem of deciding whether a given map $\varphi:G\rightarrow M$ comes from an embedding. A map $\varphi:G\rightarrow M$ is a \textbf{weak embedding} if it can be perturbed into an embedding $\psi_\eps:G\rightarrow M$ with $\|\varphi-\psi_\eps\|<\eps$ for every $\eps>0$, where $\|.\|$ is the unform norm.
		
		A polynomial-time algorithm for recognizing weak embeddings has recently been found by Fulek and Kyn\v{c}l. It reduces the problem to solving a system of linear equations over $\mathbb{Z}_2$. It runs in $O(n^{2\omega})\leq O(n^{4.75})$ time, where $\omega\in [2,2.373)$ is the matrix multiplication exponent and $n$ is the number of vertices and edges of $G$.
		We improve the running time to $O(n\log n)$. Our algorithm is also conceptually simpler: We perform a sequence of \emph{local operations} that gradually ``untangles'' the image $\varphi(G)$ into an embedding $\psi(G)$, or reports that $\varphi$ is not a weak embedding. It combines local constraints on the orientation of subgraphs directly, thereby eliminating the need for solving large systems of linear equations.
	\end{abstract}

	\section{Introduction}
	\label{sec:intro}
	
	Given a graph $G$ and a 2-dimensional manifold $M$, one can decide in linear time whether $G$ can be embedded into $M$~\cite{Moh99}, although finding the smallest genus of a surface into which $G$ embeds is NP-hard~\cite{Tho89}. An \textbf{embedding} $\psi:G\rightarrow M$ is a continuous piecewise linear injective map where the graph $G$ is considered as a 1-dimensional simplicial complex. Equivalently, an embedding maps the vertices into distinct points and the edges into interior-disjoint Jordan arcs between the corresponding vertices that do not pass through the images of vertices.
	We would like to decide whether a given map $\varphi:G\rightarrow M$ can be ``perturbed'' into an embedding $\psi:G\rightarrow M$. Let $M$ be a 2-dimensional manifold equipped with a metric.
	A continuous piecewise linear map $\phi:G\rightarrow M$ is a \textbf{weak embedding} if, for every $\eps>0$, there is an embedding $\psi_\eps:G\rightarrow M$ with $\|\varphi-\psi_\eps\|<\eps$, where $\|.\|$ is the uniform norm (i.e., $\sup$ norm).
	
	In some cases, it is easy to tell whether $\varphi:G\rightarrow M$ is a weak embedding: Every embedding is a weak embedding; and if $\varphi$ maps two edges into Jordan arcs that cross transversely, then $\varphi$ is not a weak embedding. The problem becomes challenging when $\varphi$ maps several vertices (edges) to the same point (Jordan arc), although no two edges cross transversely. This scenario arises in applications in clustering, cartography, and visualization, where nearby vertices and edges are often bundled to the same point or overlapping arcs, due to data compression, graph semantics, or low resolution. A \emph{cluster} in this context is a subgraph of $G$ mapped by $\varphi$ to the single point in $M$.
	
	The recognition of weak embeddings turns out to be a purely combinatorial problem independent of the global topology of the manifold $M$ (as noted in~\cite{ChEX15}). The key observation here is that we are looking for an embeddding in a small neighborhood of the image $\phi(G)$, which can be considered as the embedding of some graph $H$. As such, we can replace $M$ with a neighborhood of an embedded graph in the formulation of the problem.
	
	\begin{figure}[h]
		\centering
		\def\svgwidth{.8\textwidth}
		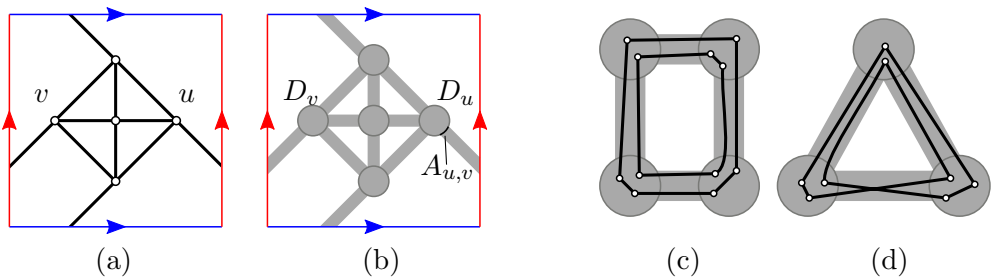
		\caption{\small
			(a) An embedding of $H=K_5$ in the torus.
			(b) Strip system $\mathcal{H}$ of the embedding of $H$.
			(c) A weak embedding where $G$ is disconnected and $H=C_4$.
			(d) A negative instance where $G=C_8$ and $H=C_3$.}
		\label{fig:intro}
	\end{figure}
	
	\medskip\noindent\textbf{Problem Statement and Results.}
	In Sections~\ref{sec:pre}--\ref{sec:reverse}, we make two simplifying assumptions: We assume that $M$ is an \textbf{orientable} 2-manifold, and $\varphi:G\rightarrow H$ is a \emph{simplicial map} for some graph $H$ embedded in $M$. A map $\varphi:G\rightarrow H$ is \textbf{simplicial} if it maps the vertices of $G$ to vertices of $H$ and the edges of $G$ to edges or vertices of $H$ such that incidence relations are preserved (in particular, the image of an edge cannot pass through the image of a vertex). We explain below how to drop the assumption that $\varphi:G\rightarrow H$ is simplicial at the expense of increasing the running time of our algorithms. We extend our results to nonorientable surfaces in Section~\ref{sec:nonorientable}.
	
	An \textbf{embedded graph} $H$ in an orientable 2-manifold $M$ is an abstract graph together with a \textbf{rotation system} that specifies, for each vertex of $H$, the ccw cyclic order of incident edges.
	The \textbf{strip system} $\mathcal{H}$ of $H$ (a.k.a. the \emph{thickening} of $H$) is a 2-manifold with boundary constructed as follows (see Fig.~\ref{fig:intro}(a)--(b)): For every $u\in V(H)$, create a topological disk $D_u$, and for every edge $uv\in E(H)$, create a rectangle $R_{uv}$. For every $D_u$ and $R_{uv}$, fix an arbitrary orientation of the boundaries $\partial D_u$ and $\partial R_{uv}$, respectively. Partition $\partial D_u$ into $\deg(u)$ arcs, and label them by $A_{u,v}$, for all $uv\in E(H)$, in the cyclic order around $\partial D_u$ determined by the rotation of $u$ in the rotation system of $H$. Finally, the manifold $\mathcal{H}$ is obtained by identifying two opposite sides of every rectangle $R_{uv}$ with $A_{u,v}$ and $A_{v,u}$ via an orientation preserving homeomorphism (i.e., consistently with the chosen orientations of $\partial R_{uv},\partial D_u$ and $\partial D_v$). 
	%
	If $\varphi:G\rightarrow M$ is a continuous piecewise linear map, then $\varphi(G)$ is the embedding of some graph $H$ in $M$, and an $\eps$-neighborhood of $\varphi(G)$ in $M$ is homeomorphic to $\mathcal{H}$ for a sufficiently small $\eps>0$. Consequently, $\varphi:G\rightarrow M$ is a weak embedding if and only if $\varphi:G\rightarrow \mathcal{H}$ is a weak embedding. We may further assume that $\psi_\eps$ maps every edge $uv$ to a Jordan arc that crosses the boundaries of $D_u$ (resp., $D_v$) precisely once (cf.~\cite[Lemma~B2]{ChEX15} and~\cite[Section~4]{FK18_ht}).
	
	We can now formulate an instance of the weak embeddability problem as a simplicial map $\varphi:G\rightarrow H$ (for short, $\varphi$), where $G$ is an abstract graph and $H$ is an embedded graph. The simplicial map $\varphi:G\rightarrow H$ is a \textbf{weak embedding} if there exists an embedding $\psi_\varphi:G\rightarrow \mathcal{H}$ that maps each vertex $v\in V$ to a point in $D_{\varphi(v)}$, and each edge $uv\in E(G)$ to a Jordan arc in $D_{\varphi(u)}\cup R_{\varphi(u)\varphi(v)}\cup D_{\varphi(v)}$ that has a connected intersection with each of $D_{\varphi(u)}$, $R_{\varphi(u)\varphi(v)}$, and $D_{\varphi(v)}$, and $R_{\varphi(u)\varphi(v)}=\emptyset$ if $u=v$. We say that the embedding $\psi_{\varphi}$ \textbf{approximates} $\varphi$. Our main results is the following.
	
	\begin{theorem}\label{thm:main}
		\begin{enumerate}
			\item[]
			\item[{\rm (i)}] Given an abstract graph $G$ with $m$ edges, an embedded graph $H$, and a simplicial map $\varphi: G\rightarrow H$, we can decide in $O(m\log m)$ time whether $\phi$ is a weak embedding.
			\item[{\rm (ii)}] If $\varphi:G\rightarrow H$ is a weak embedding, then for every $\eps>0$ we can also find an embedding
			$\psi_\eps:G\rightarrow M$ with $\|\varphi-\psi_\eps\|<\eps$ in $O(m\log m)$ time.
		\end{enumerate}
	\end{theorem}
	Throughout the paper we assume that $G$ has $n$ vertices and $m$ edges.
	In the plane (i.e., $M=\mathbb{R}^2$), only planar graphs admit weak embeddings hence $m=O(n)$, but our techniques work for 2-manifolds of arbitrary genus, and $G$ may be a dense graph. Our result improves the running time of the previous algorithm~\cite{FK18_ht} from $O(m^{2\omega})\leq O(m^{4.75})$ to $O(m\log m)$, where $\omega\in [2,2.373)$ is the matrix multiplication exponent. It also improves the running times of several recent polynomial-time algorithms in special cases, e.g., when the embedding of $G$ is restricted to a given isotopy class~\cite{fulek:LIPIcs:2017:8223}, and $H$ is a path~\cite{ADDF17} (see below).
	
	\medskip\noindent\textbf{Extension to Nonsimplicial Maps and Nonorientable Manifolds.}
	If $\varphi:G\rightarrow H$ is a continuous map (not necessarily simplicial) that is injective on the edges (each edge is a Jordan arc), we may assume that $\varphi(V(G))\subseteq V(H)$ by subdividing the edges in $E(H)$ with at most $n=|V(G)|$ new vertices in $V(H)$ if necessary. Then $\varphi$ maps every edge $e\in G$ to a path of length $O(n)$ in $H$. By subdividing the edges $e\in E(G)$ at all vertices in $V(H)$ along $\varphi(e)$, we reduce the recognition problem to the regime of simplicial maps (Theorem~\ref{thm:main}). The total number of vertices of $G$ may increase to $O(mn)$ and the running time to $O(mn\log (mn))=O(mn\log (m))$.
	\begin{corollary}\label{cor:nonsimplicial}
		Given an abstract graph $G$ with $m$ edges and $n$ vertices, an embedded graph $H$,
		and a piecewise linear continuous map $\varphi: G\rightarrow H$ that is injective on the
		interior of every edge in $E(G)$, we can decide in $O(mn\log (m))$ time whether $\phi$ is a weak embedding.
	\end{corollary}
	For example, this applies to straight-line drawings in $\mathbb{R}^2$ if the edges may pass through vertices.
	\begin{corollary}\label{cor:geo}
		Given an abstract graph $G$ with $n$ vertices and a map $\varphi: G\rightarrow \mathbb{R}^2$
		where every edge is mapped to a straight-line segment, we can decide in
		$O(n^2 \log n)$ time whether $\phi$ is a weak embedding.
	\end{corollary}
	
	In Section~\ref{sec:nonorientable}, we extend Theorem~\ref{thm:main} and Corollary~\ref{cor:nonsimplicial} to \textbf{nonorientable} surfaces with minor changes in the combinatorial representations, using a signature $\lambda: E(H)\rightarrow\{-1,1\}$ to indicate whether an edge $uv$ (and $R_{uv}$) is orientation-preserving or -reversing (similarly to~\cite{BD09}).
	
	\medskip\noindent\textbf{Related Previous Work.}
	The study of weak embeddings lies at the interface of several independent lines of research in mathematics and computer science. In topology, the study of weak embeddings and its higher dimensional analogs were initiated by Sieklucki~\cite{S69_realization} in the 1960s.
	One of his main results~\cite[Theorem 2.1]{S69_realization} implies the following. Given a graph $G$ and an embedded path $H$, \emph{every} simplicial map $\varphi: G \rightarrow H$ is a weak embedding if and only if every connected component of $G$ is a subcubic graph with at most one vertex of degree three.
	However, even if $H$ is a cycle or a 3-star, it is easy to construct simplicial maps $\varphi: G \rightarrow H$, where $G$ has maximum degree 2, that are not weak embeddings. 
	In this case, a series of recent papers on weakly simple \emph{polygons}~\cite{AAET17,ChEX15,CDPP09} show that weak embeddings can be recognized in $O(n\log n)$ time (the same time suffices when $\phi:C_n\rightarrow H$ is a simplicial map~\cite{ChEX15} or a continuous map~\cite{AAET17}).
	
	Finding efficient algorithms for the recognition of weak embeddings $\varphi: G \rightarrow H$, where $G$ is an arbitrary graph, was posed as an open problem in~\cite{AAET17,ChEX15,CDPP09}. The first polynomial-time solution for the general version follows from a recent variant~\cite{FK18_ht} of the Hanani-Tutte theorem~\cite{Ha34_uber,Tutte70_toward}, which was conjectured by M.~Skopenkov~\cite{Sko03_approximability} in 2003 and in a slightly weaker form already by Repov\v{s} and A.~Skopenkov~\cite{ReSk98_deleted} in 1998. However, this algorithm reduces the problem to a system of $O(m)$ linear equations over $\mathbb{Z}_2$, where $m=|E(G)|$. The running time is dominated by solving this system in $O(m^{2\omega})\leq O(m^{4.75})$ time, where $\omega\in [2,2.373)$ is the matrix multiplication exponent; cf.~\cite{WVV12}.
	
	Weak embeddings of graphs also generalize various graph visualization models such as the recently introduced \emph{strip planarity}~\cite{ADDF17} and \emph{level planarity}~\cite{JLM98_level}; and can be seen as a special case~\cite{AngeliniL19} of the notoriously difficult \emph{clustered-planarity} (for short, \emph{c-planarity})~\cite{CorteseP18,FCE95a_how,FCE95b_planarity}, whose tractability remains elusive
	despite many attempts by leading researchers.
	
	\medskip\noindent\textbf{Outline.}
	Our results rely on ideas from~\cite{ADDF17,ChEX15,CDPP09} and~\cite{FK18_ht}.
	To distinguish the graphs $G$ and $H$, we use the convention that $G$ has \textbf{vertices} and \textbf{edges}, and $H$ has \textbf{clusters} and \textbf{pipes}. A cluster $u\in V(H)$ corresponds to a subgraph $\varphi^{-1}[u]$ of $G$,
	and a pipe $uv\in E(H)$ corresponds to a set of edges $\varphi^{-1}[uv]\subseteq E(G)$.
	
	The main tool in our algorithm is a local operation, called ``cluster expansion,'' which generalizes similar operations introduced previously for the case that $G$ is a cycle~\cite{CDPP09}. Given an instance $\varphi:G\rightarrow H$ and a cluster $u\in V(H)$, it modifies $u$ and its neighborhood (by replacing $u$ with several new clusters and pipes) such that weak embeddability is invariant under the operation in the sense that the resulting new instance $\varphi':G'\rightarrow H'$ is a weak embedding if and only if $\varphi:G\rightarrow H$ is a weak embedding. Our operation increases the number of clusters and pipes, but it decreases the number of ``ambiguous'' edges (i.e., multiple edges in the same pipe). The proof of termination and the running time analysis use potential functions.
	
	In a preprocessing phase, we perform a cluster expansion operation at each cluster $u\in V(H)$. The main loop of the algorithm applies another operation, ``pipe expansion,'' for two adjacent clusters $u,v\in V(H)$ under certain conditions. It merges the clusters $u$ and $v$, and the pipe $uv\in E(H)$ between them, and then invokes cluster expansion. If any of these operations finds a local configuration incompatible with an embedding, then the algorithm halts and reports that $\varphi$ is not a weak embedding (this always corresponds to some nonplanar subconfiguration since the neighborhood of a single cluster or pipe is homeomorphic to a disk). We show that after $O(m)$ successive operations, we obtain an irreducible instance for which our problem is easily solvable in $O(m)$ time. Ideally, we end up with $G=H$ (one vertex per cluster and one edge per pipe), and $\varphi=\mathrm{id}$ is clearly an embedding. Alternatively, $G$ and $H$ may each be a cycle (possibly $G$ winds around $H$ multiple times), and we can decide whether $\varphi$ is a weak embedding in $O(m)$ time by a simple traversal of $G$.
	If $G$ is disconnected, then each component falls into one of the above two cases (Fig.~\ref{fig:intro}(c)--(d)), i.e., the case when $\varphi=\mathrm{id}$ or the case when $\varphi\not=\mathrm{id}$.
	
	The main challenge was to generalize previous local operations (that worked well for  cycles~\cite{ChEX15,CDPP09,Sko03_approximability}) to arbitrary graphs. Our expansion operation for a cluster $u\in V(H)$
	simplifies each component of the subgraph $\varphi^{-1}[u]$ of $G$ independently. Each component is planar (otherwise it cannot be perturbed into an embedding in a disk $D_u$). However, a planar (sub)graph with $k$ vertices may have $2^{\Omega(k\log k)}$ combinatorially different embeddings: some of these may or may not be compatible with adjacent clusters. The embedding of a (simplified) component $C$ of $\varphi^{-1}[u]$ depends, among other things, on the edges that connect $C$ to adjacent clusters. The \textbf{pipe-degree} of $C$ is the number of pipes that contain the edges incident to $C$.
	If the pipe-degree of $C$ is 3 or higher, then the rotation system of $H$ constrains the embedding of $C$. If the pipe-degree is 2, however, then the embedding of $C$ can only be determined up to a reflection, unless $C$ is connected by \emph{two} independent edges to a component in $\varphi^{-1}[v]$ whose orientation is already fixed; see Fig.~\ref{fig:flip}.
	
	\begin{figure}[h]
		\centering
		\def\svgwidth{.75\textwidth}
		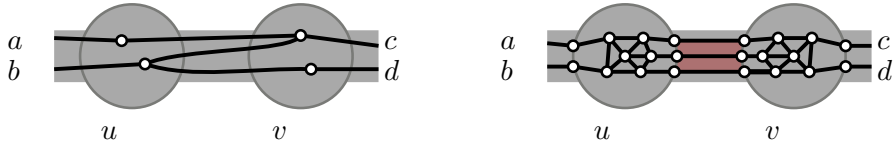
		\caption{\small
			Two adjacent clusters, $u$ and $v$, that each contain two components of pipe-degree 2 (left).
			These components merge into a single component in $D_u\cup R_{uv}\cup D_v$.
			In every embedding, the order of the pipe-edges $a,b$ determines the order of the pipe-edges $c,d$.
			The operation \textsf{pipeExpansion}$(uv)$ transforms the component on the left to two wheels
			connected by three edges (a so-called thick edge) shown on the right.
		}
		\label{fig:flip}
	\end{figure}
	
	We need to maintain the feasible embeddings of the components in all clusters efficiently. In~\cite{FK18_ht}, this problem was resolved by introducing 0-1 variables for the components, and aggregating the constraints into a system of linear equations over $\mathbb{Z}_2$, which was eventually resolved in $O(m^{2\omega})\leq O(m^{4.75})$ time. We improve the running time to $O(m\log m)$ by maintaining the feasible embeddings simultaneously with our local operations.
	
	Another challenge comes from the simplest components in a cluster $\varphi^{-1}[u]$. Long chains of degree-2 vertices, with one vertex per cluster, are resilient to our local operations. Their length may decrease by only one (and cycles are irreducible). We need additional data structures to handle these ``slowly-evolving'' components efficiently. We use a dynamic heavy-path decomposition data structure and a suitable potential function to bound the time spent on such components.
	
	\medskip\noindent\textbf{Organization.}
	In Section~\ref{sec:pre}, introduce additional terminology for an instance $\varphi:G\rightarrow H$, and show how to modify the subgraphs of $G$ within each cluster to reduce the problem to an instance in ``normal form,'' a ``simplified form,'' and introduce a combinatorial representation of weak embeddings that we use in our algorithm.
	The simplification step relies on the concept of SPQR-trees, developed in~\cite{DT96} for the efficient representation of combinatorial embeddings of a graph, which we also review in this section.
	In Section~\ref{sec:op}, we present the cluster expansion and pipe expansion operations and prove that  weak embeddibility is invariant under both operations. We use these operations repeatedly in Section~\ref{sec:alg} to decide whether a simplicial map $\varphi:G\rightarrow H$ is a weak embedding.
	Section~\ref{sec:reverse} discusses how to reverse a sequence of operations to perturb a weak embedding into an embedding.
	The adaptation of our results to nonorientable surfaces $M$ is discussed in Section~\ref{sec:nonorientable}.
	We conclude with open problems in Section~\ref{sec:con}.

	\section{Preliminaries}
	\label{sec:pre}
	
	In this section, we describe modifications within the clusters of a simplicial map $\varphi:G\rightarrow H$ to bring it to ``normal form'' (properties \ref{P1}--\ref{P2}) and ``simplified form'' (properties \ref{P3}--\ref{P4}). These properties allow for a purely combinatorial representation of weak embeddability (in terms of permutations), which we use
	in the proof of correctness of the algorithms in Sections~\ref{sec:op} and~\ref{sec:alg}.
	
	\paragraph{Definitions.}
	Two instances $\phi:G\rightarrow H$ and $\phi':G'\rightarrow H'$ are called \textbf{equivalent} if $\phi$ is a weak embedding if and only if $\phi'$ is a weak embedding.
	We call an edge $e\in E(G)$ a \textbf{pipe-edge} if $\varphi(e)\in E(H)$, or a \textbf{cluster-edge} if $\varphi(e)\in V(H)$.
	For every cluster $u\in V(H)$, let $G_u$ be the subgraph of $G$ induced by $\phi^{-1}[u]$.
	For every pipe $uv\in E(H)$, $\varphi^{-1}[uv]$ stands for the set of pipe-edges mapped to $uv$ by $\varphi$.
	
	The \textbf{pipe-degree} of a connected component $C$ of $G_u$, denoted $\text{pipe-deg}(C)$, is the number of pipes that contain some edge of $G$ incident to $C$.
	A vertex $v$ of $G_u$ is called a \textbf{terminal} if it is incident to a pipe-edge.
	For an integer $k\geq 3$, the $k$-vertex \textbf{wheel} graph $W_k$ is a join of a \textbf{center} vertex $c$ and a cycle of $k-1$ \textbf{external} vertices.
	Refer to \cite{D16_graph_theory} for standard graph theoretic terminology (e.g., cut vertex, 2-cuts, biconnectivity).

	\subsection{Normal Form}
	An instance $\phi:G\rightarrow H$ is in \textbf{normal form} if every cluster $u\in V(H)$ satisfies:
	
	\begin{enumerate}[label=(P\arabic*)]
		\item \label{P1} Every terminal in $G_u$ is incident to exactly one cluster-edge and one pipe-edge.
		\item \label{P2} There are no degree-2 vertices in $G_u$.
	\end{enumerate}
	
	We now describe subroutine \textsf{normalize}$(u)$ that, for a given instance $\varphi:G\rightarrow H$ and a cluster $u\in V(H)$,
	returns an equivalent instance $\varphi':G'\rightarrow H$ such that $u$ satisfies \ref{P1}--\ref{P2}; refer to Fig.~\ref{fig:normalized}(a)--(b).
	
	\begin{figure}[h]
		\centering
		\def\svgwidth{\textwidth}
		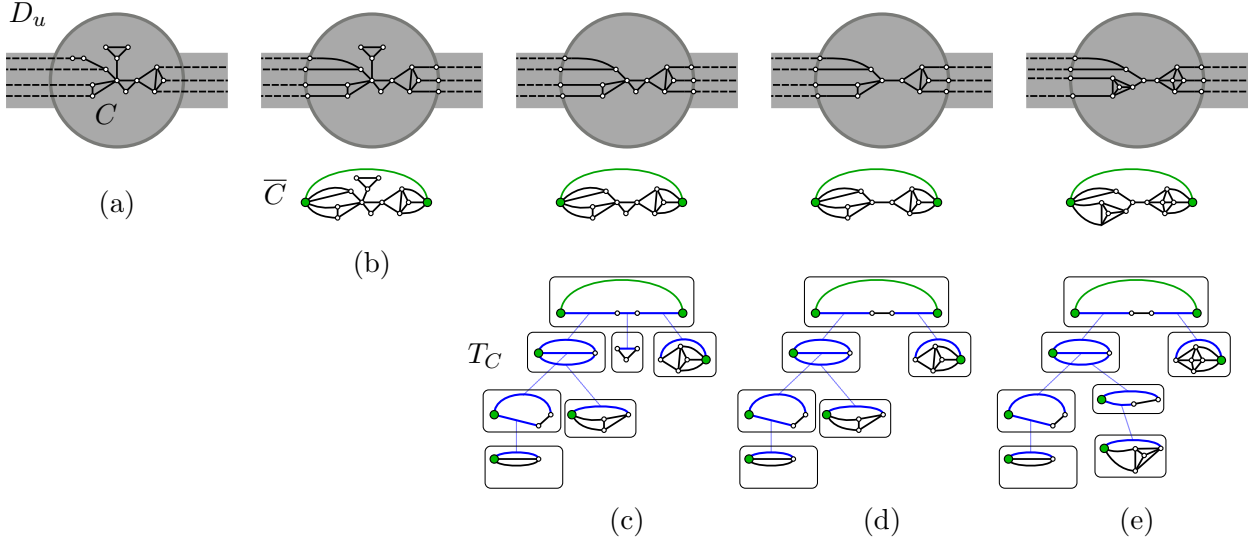
		\caption{\small
			Changes in a cluster caused by \textsf{normalize} and \textsf{simplify}. (a) Input, (b) after \textsf{normalize}, (c) after first part of step 1, (d) after step 1, and (e) after step 2 of subroutine \textsf{simplify}.
			Dashed lines, green dots, green lines, and blue lines represent pipe-edges, pipe-vertices, edges in $\overline{E}_C$, and virtual edges, respectively.}
		\label{fig:normalized}
	\end{figure}
	
	\paragraph{\textsf{normalize}$(u)$.} Input: an instance $\varphi:G\rightarrow H$ and a cluster $u\in V(H)$.
	
	\smallskip\noindent
	Subdivide every pipe-edge $pq$ where $\phi(p)=u$ into a path $(p,p',q)$ such that $\phi'(p)=\phi'(p')=u$, and $\phi'(q)=\phi(q)$.
	Note that the new vertex $p'$ is a terminal in $G$ and a leaf in $G_u$ (i.e., $\deg_{G_u}(p')=1$).
	Successively suppress every vertex $p$ of $G_u$ with $\deg_{G_u}(p)=2$ by merging its incident edges.
	If this creates a loop, delete the loop.
	
	\begin{lemma}\label{lem:normalize}
		Given an instance $\phi:G\rightarrow H$ and a cluster $u\in V(H)$, the instance $\phi'=\textsf{normalize}(u)$ and $\phi$ are equivalent and $u$ satisfies \ref{P1}--\ref{P2} in $\phi'$.
		The subroutine runs in $O(\sum_{p\in V(G_u)} \deg_{G}(p))$ time.
		By successively applying $\textsf{normalize}$ to all clusters in $V(H)$, we obtain an equivalent instance in normal form in $O(|E(G)|)$ time.
	\end{lemma}
	\begin{proof}
		The instances $\phi$ and $\phi'$ are clearly equivalent since
		(i) we can always replace the embedding of an edge by a path and vice-versa, and
		(ii) a loop can always be deleted or added to a vertex in an embedding.
		There are $O(\sum_{p\in V(G_u)} \deg_{G}(p))$ pipe-edges incident to a vertex in $u$.
		Hence the subroutine performs $O(\sum_{p\in V(G_u)} \deg_{G}(p))$ subdivisions.
		There are at most $|E(G_u)|=O(\sum_{p\in V(G_u)} \deg_{G}(p))$ degree-2 vertices in $G_u$.
		By construction, the resulting graph $G_u'$ satisfies \ref{P1}--\ref{P2}.
		All changes are local and applying $\textsf{normalize}$ to $u$ does not
		change properties \ref{P1}--\ref{P2} in other clusters.
		Therefore, we can obtain the normal form of $\phi$ in $O(|E(G)|)$ time.
	\end{proof}
	
	For every cluster $u\in V(H)$, the graph $G_u$ may have several components. For each component $C$, we define a multigraph $\overline{C}$ that represents the interactions of $C$ with vertices in adjacent clusters. Specifically,
	for each component $C$ of $G_u$ of a cluster $u\in V(H)$ satisfying \ref{P1}--\ref{P2}, we define the multigraph $\overline{C}$ in two steps as follows; refer to Fig.~\ref{fig:normalized}(b).
	
	\begin{enumerate}
		\item For every pipe $uv\in E(H)$ incident to $u$, create a new vertex $v'$, called \textbf{pipe-vertex}, and identify all terminal vertices of
		$C$ incident to some edge in $\varphi^{-1}[uv]$ with $v'$ (this may create multiple edges incident to $v'$).
		
		\item If $\text{pipe-deg}(C)=2$, connect the two pipe-vertices with an edge $e$ and let $\overline{E}_C=\{e\}$.
		If $\text{pipe-deg}(C)\geq 3$, connect all pipe-vertices by a cycle in the order determined by the rotation of $u$
		and let $\overline{E}_C$ be the set of edges of this cycle. Hence, if 
		$\text{pipe-deg}(C)\le 1$ then $\overline{E}_C=\emptyset$.
	\end{enumerate}

	Let $\overline{C}=(V(\overline{C}),E(\overline{C}))$, where
	$V(\overline{C})$ consists of nonterminal vertices in $V(C)$ and $\text{pipe-deg}(C)$ pipe-vertices,
	and $E(\overline{C})$ consists of the edges induced by nonterminal vertices in $V(C)$, (multi) edges created in step~1 (each of which corresponds to an edge in $E(C)$) and edges in $\overline{E}_C$.
	It is clear that every embedding of $\overline{C}$ can be converted into an embedding of $C$
	such that the rotation of the pipe vertices in $\overline{C}$ determines the cyclic order of
	terminals along the facial walk of the outer face of $C$.

	\subsection{SPQR-Trees}
	\textbf{SPQR-trees} were introduced by Di Battista and Tamassia~\cite{DT96} for an efficient representation of all combinatorial plane embeddings of a graph. Let $G$ be a biconnected planar graph. The SPQR-tree $T_G$ of $G$ represents a recursive decomposition of $G$ defined by its (vertex) 2-cuts.
	A deletion of a 2-cut $\{u,v\}$ disconnects $G$ into two or more components $C_1,\ldots ,C_i$, $i\geq 2$.
	A \textbf{split component} of $\{u,v\}$ is either an edge $uv$ (which is not one of the components $C_j$) or the subgraphs of $G$ induced by $V(C_j)\cup\{u,v\}$ for $j=1,\ldots, i$.
	The tree $T_G$ captures the recursive decomposition of $G$ into split components defined by 2-cuts of $G$.
	A node $\mu$ of $T_G$ is associated with a multigraph called $\skel(\mu)$ on a subset of $V(G)$, and has a \textbf{type} in $\{$S,P,R$\}$. If the type of $\mu$ is S then $\skel(\mu)$ is a cycle of 3 or more vertices.
	If the type of $\mu$ is P then  $\skel(\mu)$ consists of 3 or more parallel edges between a pair of vertices. If the type of $\mu$ is R
	then  $\skel(\mu)$ is a 3-connected graph on 4 or more vertices.
	An edge in $\skel(\mu)$ is \textbf{real} if it is an edge in $G$, or \textbf{virtual} otherwise.
	A virtual edge connects the two vertices of a 2-cut, $u$ and $v$, and represents a subgraph of $G$
	obtained in the recursive decomposition, containing a $uv$-path in $G$ that does not contain any edge in $\skel(\mu)$.
	Two nodes $\mu_1$ and $\mu_2$ of $T_G$ are adjacent if $\skel(\mu_1)$ and $\skel(\mu_2)$ share exactly two vertices, $u$ and $v$, that form a 2-cut in $G$.
	Each virtual edge in $\skel(\mu)$ corresponds bijectively to a pair of adjacent nodes in $T_G$.
	No two S nodes (resp., no two P nodes) are adjacent. Every edge in $E(G)$ appears in the skeleton of exactly one node.
	The tree $T_G$ has $O(|E(G)|)$ nodes and it can be computed in $O(|E(G)|)$ time~\cite{DT96}.
	
	It is also known that $T_G$ represents all combinatorial embeddings of $G$ in $\mathbb{R}^2$
	in the following manner~\cite{DT96}. Choose a root node for $T_G$ and an embedding of its skeleton.
	Then successively replace each virtual edge $uv$ by the skeleton of the corresponding node $\mu$
	minus the virtual edge $uv$ in $\skel(\mu)$. In each step of the recursion,
	if $\mu$ is of type R, $\skel(\mu)$ can be flipped (reflected) around $u$ and $v$,
	and if $\mu$ is of type P, the parallel edges between $u$ and $v$ can be permuted arbitrarily.
	
	\subsection{Combinatorial Representation of Weak Embeddings}
	\label{sec:combRep}
	Given an embedding $\psi_\phi:G\rightarrow\mathcal{H}$, where $\phi$ is in normal form, we define a combinatorial representation $\pi=\pi(\psi_\phi)$ as the set of total orders of edges in $R_{uv}$, for all pipes $uv\in E(H)$. Specifically, for every pipe $uv\in E(H)$, fix an orientation of the boundary of $R_{uv}$ (e.g., the one used in the construction of the strip system $\mathcal{H}$). Record the order in which the pipe-edges in the embedding $\psi_\phi$ intersect a fixed side of $\partial R_{uv}$ (say, the side $\partial R_{uv}\cap \partial D_u$) when we traverse it in the given orientation. Let $\Pi(\varphi)$ be the set of combinatorial representations $\pi(\psi_\phi)$ of all embeddings $\psi_\phi:G\rightarrow\mathcal{H}$. Note that $\varphi:G\rightarrow H$ is a weak embedding if and only if $\Pi(\varphi)\neq\emptyset$.
	
	Conversely, let $\pi$ be a set of total orders of edges in rectangles $R_{uv}$, for all pipes $uv\in E(H)$. We show (in Lemma~\ref{lem:comb-reconstruction} below) how to use the normal form and SPQR-trees to decide whether $\pi\not\in\Pi(\varphi)$ .
	We say that two components $C_1$ and $C_2$ of $G_u$ \textbf{cross} with respect to $\pi$ if and only if their terminals interleave in the cyclic order around $\partial D_u$ (i.e., there exists no cut in the cyclic order in which all terminals of $C_1$ appear before all terminals of $C_2$). If $\pi\in\Pi(\varphi)$, then $\pi$ cannot induce two crossing components in $G_u$, for any $u\in V(H)$.
	
	\begin{lemma}\label{lem:comb-reconstruction}
		Given a set of total orders $\pi$ for a map $\phi:G\rightarrow H$, we can decide whether $\pi\not\in\Pi(\varphi)$ in $O(m)$ time.
		If $\pi\in\Pi(\varphi)$, then we can also find an embedding $\psi_\phi:G\rightarrow\mathcal{H}$, $\pi=\pi(\psi_\phi)$, in $O(m)$ time.
	\end{lemma}
	\begin{proof}
		Let $\mathcal{H}$ be the strip system for $\varphi: G\rightarrow H$.
		For each pipe $uv\in E(H)$, draw $|\phi^{-1}[uv]|$ parallel Jordan arcs in $R_{uv}$ connecting $\partial D_u$ and $\partial D_v$. For each cluster $u\in V(H)$, $\pi$ defines a ccw cyclic order of terminals around $\partial D_u$.
		Let $t_C$ denote the number of terminals in $C$.
		Create a graph $\widetilde{C}$ by the union of $C$ and a wheel $W_{t_C+1}$ whose center is a new vertex, and whose external vertices are the terminals of $C$ in the order defined by $\pi$.
		An embedding of $C$ in which the terminals appear in the outer face in the same cyclic order as the one defined by $\pi$ exists if and only if $\widetilde{C}$ is planar. If $\widetilde{C}$ is not planar, report that $\pi\notin\Pi(\varphi)$. Else, embed $C$ inside $D_u$ given the position of the already embedded terminals on $\partial D_u$. This subdivides $D_u$ into faces. If a component connects terminals in different faces, then two components in $G_u$ cross, hence $\pi\notin\Pi(\varphi)$. If no two components cross, we can incrementally embed all components of $G_u$, as there is always a face of $D_u$ that contains all terminals of each remaining component.
	\end{proof}

	\subsection{Simplified Form}
	Given an instance $\phi:G\rightarrow H$ in normal form, we simplify the graph $G$ by removing parts of $G_u$, for all $u\in V(H)$, that are locally ``irrelevant'' for the embedding, such as 0-, 1-, and 2-connected components that are not adjacent to edges in any pipe incident to $u$.
	Formally, for each component $C$ of $G_u$, we call a split component defined by a 2-cut $\{p,q\}$ of $\overline{C}$ \textbf{irrelevant} if it contains no pipe-vertices.
	An instance is in \textbf{simplified form} if it is in normal form an every $u\in V(H)$ satisfies properties \ref{P3}--\ref{P4} below.
	
	\begin{enumerate}[label=(P\arabic*)]
		\setcounter{enumi}{2}
		\item \label{P3} For every component $C$ of $G_u$, $\overline{C}$ is biconnected and
		every 2-cut of $\overline{C}$ contains at least one pipe-vertex.	
	\end{enumerate}
	
	Assuming that a cluster $u$ satisfies \ref{P3}, we define $T_C$ as the SPQR tree of $\overline{C}$ where $C$ is a component of $G_u$.
	Given a node $\mu$ of $T_C$, let the \emph{core} of $\mu$, denoted $\core(\mu)$, be the subgraph obtained from $\skel(\mu)$ by deleting all pipe-vertices.
	Property \ref{P4} below will allow us to bound the number of vertices of $G_u$ in terms of its number of terminals (cf.~Lemma~\ref{obs:1}).

	\begin{enumerate}[label=(P\arabic*)]
		\setcounter{enumi}{3}
		
		\item \label{P4} For every component $C$ of $G_u$, and every R node $\mu$ of $T_C$, $\core(\mu)$ is
		isomorphic to a wheel $W_k$, for some $k\geq 4$, whose external vertices have degree 4 in $G_u$.
		
	\end{enumerate}

	We now describe subroutine \textsf{simplify}$(u)$ that, for a given instance $\varphi:G\rightarrow H$ in normal form and a cluster $u\in V(H)$,
	returns an instance $\varphi':G'\rightarrow H$ such that $u$ satisfies \ref{P1}--\ref{P4}.
	We break the subroutine into two steps.
	
	\paragraph{\textsf{simplify}$(u)$.} Input: an instance $\varphi:G\rightarrow H$ in normal form and a cluster $u\in V(H)$.
	
	\noindent
	For every component $C$ of $G_u$, do the following.
	
	\smallskip\noindent\textbf{(1)}
	If $C$ is not planar, report that $\phi$ is not a weak embedding and halt.
	If $\text{pipe-deg}(C)=0$, then delete $C$. Else compute $\overline{C}$, and find the maximal biconnected component
	$\widehat{C}$ of $\overline{C}$ that contains all pipe-vertices.
	The component $\widehat{C}$ trivially exists if $\text{pipe-deg}(C)\in\{1,2\}$, and if $\text{pipe-deg}(C)\geq 3$, it exists since $\overline{E}_C$ forms a
	cycle containing all pipe-vertices. Modify $C$ by deleting all vertices of $\overline{C}\setminus \widehat{C}$, and update $\overline{C}$ (by deleting the same vertices from $\overline{C}$, as well);
	refer to Fig.~\ref{fig:normalized}(b)--(c).
	Consequently, we may assume that $\overline{C}$ is biconnected and contains all pipe-vertices.
	Compute the SPQR tree $T_C$ for $\overline{C}$.
	Set a node $\mu_r$ in $T_C$ whose skeleton contain a pipe-vertex as the root of $T_C$.
	Traverse $T_C$ using DFS.
	If a node $\mu$ is found such that $\skel(\mu)$ contains no pipe-vertex, let $\{p,q\}$ be the 2-cut of $\overline{C}$ shared by $\skel(\mu)$ and $\skel(\text{parent}(\mu))$.
	Replace all irrelevant split components defined by $\{p,q\}$ by a single edge $pq$ in $C$.
	If $p$ or $q$ now have degree 2, suppress $p$ or $q$, respectively.
	Update $\overline{C}$ accordingly, and update $T_C$ to reflect the changes in $\overline{C}$ by changing $pq$ from virtual to real and possibly suppressing $p$ and/or $q$ in $\skel(\text{parent}(\mu))$, which also deletes node $\mu$ and its descendants since their skeletons contain edges in the deleted irrelevant split components; refer to Fig.~\ref{fig:normalized}(c)--(d).
	Continue the DFS ignoring deleted nodes.
	
	
	\smallskip\noindent\textbf{(2)}
	While there is an R node $\mu$ in $T_C$, of a component $C$ in $G_u$,  that does not satisfy \ref{P4}, do the following.
	Let $Y$ be the set of edges in $\skel(\mu)$
	adjacent to $\core(\mu)$ (i.e., edges between a vertex in $\core(\mu)$ and a pipe-vertex).
	Since $\mu$ is an R node, it represents a 3-connected planar graph, which has a unique combinatorial embedding (up to reflection).
	If we contract $\core(\mu)$ to a single vertex, the rotation of that vertex defines a cyclic order on $Y$.
	Let $e_i=p_iq_i$, where  $p_i\in \core(\mu)$ and $q_i$ is a pipe-vertex, be the $i$-th edge in this order. Recall that $e_i$ is an edge in $\skel(\mu)$ and therefore represents a subgraph $\overline{C}_{e_i}$ of $\overline{C}$.
	If $e_i$ is real, $\overline{C}_{e_i}$ is a single edge.
	Otherwise, $\overline{C}_{e_i}$ contains all split components defined by the 2-cut $\{p_i,q_i\}$ that do not contain $\core(\mu)$ as a subgraph.
	Do the following changes in $C$, which will incur changes in $\overline{C}$, as well. Replace $\core(\mu)$ by the wheel graph $W_{|Y|+1}$, 
	by deleting all interior edges and vertices of $\core(\mu)$,
	and inserting a new vertex adjacent to all vertices of the outer cycle. 
	If $\overline{C}_{e_i}$ is not a single edge, then replace 
	$p_i$ by two adjacent vertices, $p_i'$ and $p_i''$ where $p_i'$ 
	is in $W_{|Y|+1}$ and $p_i''$ is in $\overline{C}_{e_i}$ (that is,
	$p_i'$ is adjacent to $p_i''$ and the three neighbors of $p_i$ in the wheel, and $p_i''$ is adjacent to $p_i'$ and the remaining neighbors of $p_i$).
	%
	%
	Update $T_C$ to reflect the changes in $\overline{C}$ by updating $\skel(\mu)$ and adding an S node between $\mu$ and an adjacent node $\mu_i$ whose skeleton contained $e_i$ for each virtual edge $e_i\in Y$ where $p_i$ has not been not suppressed; refer to Fig.~\ref{fig:normalized}(d)--(e).
	
	\begin{lemma}\label{lem:simplify}
		Given an instance $\phi:G\rightarrow H$ in normal form and a cluster $u\in V(H)$, the instance $\phi'=\textsf{simplify}(u)$ and $\phi$ are equivalent and $u$ satisfies \ref{P1}--\ref{P4} in $\phi'$.
		The operation runs in $O(|E(G_u)|)$ time.
		By successively applying $\textsf{simplify}$ to all clusters in $V(H)$, we obtain an equivalent instance in simplified form in $O(m)$ time.
	\end{lemma}
	
	\begin{proof}
		First we prove that $u$ satisfies \ref{P1}--\ref{P4} in $\phi'$.
		Since $\phi$ is in normal form, $u$ satisfies \ref{P1}--\ref{P2} in $\phi$.
		By construction, $u$ still satisfies \ref{P1}--\ref{P2} in $\phi'$.
		For \ref{P3}, note that after step~1, $\overline{C}$ is
		biconnected and every node $\mu$ of $T_C$ contains a pipe-vertex in its skeleton.
		Step~2 does not change this property.
		This implies that $\core(\mu)$ contains only real edges for every node $\mu$.
		Suppose for contradiction that there is a 2-cut $\{p,q\}$ such that neither $p$ nor $q$ is a pipe-vertex.
		Then $\{p,q\}$ must be in $\core(\mu)$ where $\mu$ is an S node, or else either $p$ or $q$ would have been deleted for being in $\overline{C}\setminus \widehat{C}$.
		Then one split component of $\{p,q\}$ is a path of length two or more.
		But $G_u'$ has no degree-2 vertex by \ref{P2}, a contradiction.
		Hence, $u$ satisfies \ref{P3} in $\phi'$.
		By definition, after step~2 $u$ satisfies \ref{P4} in $\phi'$.
		
		We now show that the operation takes $O(|E(G_u)|)$ time.
		In step~1, planarity testing is done in linear time for each component $C$ of $G_u$~\cite{HoTa74_planarity}.
		We obtain $\widehat{C}$ by a DFS.
		We compute $T_C$ in $O(|E(C)|)$ time~\cite{DT96}.
		Replacing irrelevant split components by one edge can be done in $O(|E(C)|)$ overall time.
		In step 2, we can obtain a list of R nodes in $O(|E(C)|)$ time.
		The changes in step~2 are local, both in $C$ and $T_C$, and do not influence whether other R nodes satisfy \ref{P4}.
		Step~2 takes $O(|E(G_u)|)$ time overall by processing each R node sequentially. All the changes are local to $u$ and, by successively applying \textsf{simplify}, we obtain a simplified form in $O(|E(G)|)$ time.
		
		Finally, we show that $\phi$ and $\phi'$ are equivalent. Notice that there is a bijection between the terminals of $\phi$ and $\phi'$.
		We show that $\Pi(\varphi)=\Pi(\varphi')$, i.e., given $\pi\in\Pi(\varphi)$, then $\pi\in\Pi(\varphi')$ and vice versa.
		Notice that for every $\pi\in \Pi(\pi)$ and $u\in V(H)$, two components $C_1$ and $C_2$ of $G_u$ cross if and only if the corresponding components $C'_1$ and $C'_2$ of $G'_u$ in also cross.
		Then, it suffices to show that the SPQR trees of $\overline{C}$ and $\overline{C'}$ for corresponding components $C$ of $G_u$ and $C'$ of $G'_u$ represent the same constraints in the cyclic order of terminals around $D_u$.
		Step~1 deletes components of pipe-degree 0, which do not pose any restriction on the cyclic order of terminals.
		The subgraphs represented by irrelevant subtrees in the SPRQ tree of $\overline{C}$ can be flipped independently and, since they do not contain pipe-vertices, their embedding does not interfere with the order of edges adjacent to pipe-vertices.
		Hence, step~1 does not alter any constraint on the cyclic order of terminals.
		By construction, step~2 does not change the circular order of edges in $\core(\mu)$.
		Replacing $\core(\mu)$ by a wheel $W_{|Y|+1}$ does not change any of the constraints on the cyclic order of terminals.
	\end{proof}

	\begin{lemma}\label{obs:1}
		After \textsf{simplify}$(u)$, every component $C$ of $G_u$ contains $O(t_{C})$ edges,
		where $t_{C}$ is the number of terminals in $C$.
	\end{lemma}
	\begin{proof}
		By~\ref{P4}, every wheel subgraph in $C$ is a maximal biconnected component. 
		Let us contract every wheel in $C$ into a single vertex and remove any loops created by the contraction. Let $\widehat{C}$ denote the resulting component. We have $|E(C)|\le 5|E(\widehat{C})|$, since the number of edges decreases by at most
		$4|E(\widehat{C})|$.
		Indeed, a wheel $W_{k+1}$ has $2k$ edges, which are contracted, and its $k$ external vertices are incident to $k$ edges that are not contracted by~\ref{P4}. We charge each of these $k$ edges in $E(\widehat{C})$ for two edges of $W_k$. Then every edge in $E(\widehat{C})$ receives at most 2 units of charge from each of its endpoints, hence at most 4 units of charge overall.
		As all maximal biconnected components of $C$ have been contracted to single vertices of degree at least three, $\widehat{C}$ is a tree without degree-2 vertices whose leaves are precisely the terminals of $C$ by~\ref{P1}--\ref{P3}. Since the number of edges in a tree is at most twice the number its leaves, by the above inequality we have	$|E(C)|\le 5|E(\widehat{C})|\le 5\cdot 2 \cdot t_{C}$,	as claimed.
	\end{proof}

	
	\section{Operations}
	\label{sec:op}
	
	In this section, we present our two main operations, \textsf{clusterExpansion} and \textsf{pipeExpansion}, that we use successively in our recognition algorithm. Given an instance $\varphi$ and a cluster $u$ in simplified form, operation \textsf{clusterExpansion}$(u)$ either finds a configuration that cannot be embedded locally in the neighborhood of $u$ and reports that $\varphi$ is not a weak embedding, or replaces cluster $u$ with a group of clusters and pipes (in most cases reducing the number of edges in pipes). It first modifies the embedded graph $H$, and then handles each component of $G_u$ independently.
	
	Operation \textsf{pipeExpansion}$(uv)$ first merges two adjacent clusters, $u$ and $v$ (and the pipe $uv$) into a single cluster $\cluster{uv}$ and invokes \textsf{clusterExpansion}$(\cluster{uv})$. We continue with the specifics.
	
	\subsection{Cluster Expansion}
	
	For a cluster $u\in V(H)$ in an instance $\varphi:G\rightarrow H$, let the \textbf{expansion disk} $\Delta_u$ be a topological closed disk containing a single cluster $u\in V(H)$ and intersecting only pipes incident to $u$.
	
	\paragraph{\textsf{clusterExpansion}$(u)$.}
	Input: an instance $\varphi:G\rightarrow H$ in simplified form and a cluster $u\in V(H)$.
	We either report that $\varphi$ is not a weak embedding or return an instance $\varphi':G'\rightarrow H'$.
	The instance $\varphi'$ is computed incrementally: initially $\varphi'$ is a copy of $\varphi$.
	Steps 0--3 will insert new clusters and pipes into $H'$ that are within $\Delta_u$ without describing their embedding, Step 4 will determine the rotation system for the new clusters and check whether the rotation system induces any crossing between new pipes within $\Delta_u$, and Step~5 brings $\phi'$ to its simplified form.
	
	\smallskip
	\noindent\textbf{Step 0.}
	For each pipe $uv\in E(H)$ incident to $u$, subdivide $uv$ by inserting a cluster $u_v$ in $H'$ at the intersection of $uv$ and $\partial \Delta_u$. If $\deg(u)\ge 3$, then add a cycle $C_u$ of pipes through all clusters in $\partial \Delta_u$ (hence the clusters $u_v$ appear along $C_u$ in the order given by the rotation of $u$); and if $\deg(u)=2$, then add a pipe between the two clusters in $\partial \Delta_u$.
	Delete $u$ (and all incident pipes). As a result, the interior of $\Delta_u$ contains no clusters.
	
	\noindent\textbf{Step 1: Components of pipe-degree~1.}
	For each component $C$ of $G_u$ such that $\text{pipe-deg}(C)=1$, let $uv$ be the pipe to which the pipe-edges incident to $C$ are mapped to.
	Move $C$ to the new cluster $u_v$, i.e., set $\varphi'(C)=u_v$.
	(For example, see the component in $u_x$ in Fig.~\ref{fig:clusterExp}.)
	
	\noindent\textbf{Step 2: Components of pipe-degree~2.}
	For each pair of clusters $\{v,w\}$ adjacent to $u$, denote by $B_{vw}$ the set of components of $G_u$ of degree 2 adjacent to pipe-edges in $\varphi^{-1}[uv]$ and $\varphi^{-1}[uw]$. For all nonempty sets $B_{vw}$ do the following.
	
	\textbf{(a)} Insert the pipe $u_vu_w$ into $H'$ if it is not already present.
	
	\textbf{(b)} For every component $C\in B_{vw}$, do the following:
	
	\textbf{(b1)}
	Compute $\overline{C}$ (by \ref{P3}, $\overline{C}$ is biconnected).
	Compute the SPQR tree $T_C$ of $\overline{C}$.
	Set a node $\mu_r$ as the root of $T_C$ so that $\skel(\mu_r)$ contains both pipe-vertices,
	which we denote by $v'$ and $w'$ (i.e., consistently with Section~\ref{sec:pre}).
	Note that $\mu_r$ cannot be of type P, otherwise $C$ would not be connected.
	
	\textbf{(b2)} 
	If $\mu_r$ is of type S, then $E(\skel(\mu_r))\setminus \overline{E}_C$ forms a path between the pipe vertices $v'$ and $w'$, that we denote by $P$, where the first and last edges may be virtual.
	Notice that path $P$ contains at most 3 edges otherwise \ref{P2} or \ref{P3} would not be satisfied.
	If $P=(v',w')$ has length 1, then subdivide $P$ into 3 edges $P=(v',p_1,p_2,w')$. 
	If $P=(v',p,w')$ has length 2, then $\{v',p\}$ is a 2-cut in $\overline{C}$ that defines two split components,
	$C_{v}$ and $C_{w}$, containing $v'$ and $w'$, respectively. Note that it is not possible that both $v'p$ and $pw'$ are real edges, because there are no degree-2 vertices in $G_u$ by~\ref{P2}. Split $p$ into two vertices, $p_1$ and $p_2$,
	connected by an edge so that $p_1$ (resp., $p_2$) is adjacent to every vertex in $C_{v}$
	(resp., $C_w$) that was adjacent to $p$. 
	Finally, if $P=(v',p_1,p_2,w')$ has length 3, we do not modify $P$. 
	Then, in all three cases, the edge $p_1p_2$ defines an edge cut in $C$
	that splits $C$ into two components each incident to a single pipe,
	one to $uv$ and the other to $uw$. We define $\varphi'$ so that it maps each of the two components into
	$u_v$ or $u_w$ accordingly. (See the components incident to pipe $u_{v_1}u_{v_3}$ in Fig.~\ref{fig:clusterExp}.)
	
	\textbf{(b3)} 
	If $\mu_r$ is of type R, by \ref{P4}, $\core(\mu_r)$ is a wheel subgraph $W_{k}$.
	Let $k_v$ and $k_w$, where $k_v+k_w=k-1$, be the number of edges between $W_k$ and $v'$, and between $W_k$ and $w'$, respectively.
	Replace $W_{k}$ by two wheel graphs $W_{k_v+4}$ and $W_{k_w+4}$ connected by three edges so that the circular order of the edges around $v'$ and $w'$ is maintained (recall that an R node has a unique embedding).
	The triple of edges between $W_{k_v+4}$ and $W_{k_w+4}$ is called a \textbf{thick edge}.
	The thick edge defines a 3-edge-cut that splits $C$ into two components, each with a wheel graph.
	We define $\varphi'$ so that each of the two components is mapped to its respective vertex $u_v$ or $u_w$.
	(See the components incident to pipe $u_{v_1}u_w$ in Fig.~\ref{fig:clusterExp}.)

	\noindent\textbf{Step 3: Components of pipe-degree~3 or higher.}
	For all the remaining components $C$ (i.e., $\text{pipe-deg}(C)\ge 3$) of $G_u$, do the following.
	Assume $C$ is incident to pipe-edges mapped to the pipes $uv_1, uv_2,\ldots , uv_d$.
	
	\textbf{(a)}
	Compute $\overline{C}$ and its SPQR tree $T_C$ and let $v_i'$ be the pipe-vertex corresponding to terminals adjacent to edges in  $uv_i$ for $i=1,2,\ldots, d$.
	Set the node $\mu_r$ as the root of $T_C$ such that $\skel(\mu_r)$ contains the cycle $\overline{E}_C$.
	The type of $\mu_r$ is R, otherwise $C$ would be disconnected.
	
	\textbf{(b) Changes in $H'$}.
	By \ref{P4}, we have that $\core(\mu_r)$  is a wheel graph $W_{k}$, where $k-1\ge d$.
	For $j=1,2,\ldots, k-1$, let $p_j$ be the $j$-th external vertex of $W_{k}$ and $p_C$ be its central vertex.
	Create a copy of $W_{k}$ using clusters and pipes: Create a cluster $u_{p_j}$ that represents each vertex $p_j$,
	a cluster $u_C$ that represents vertex $p_C$, see Fig.~\ref{fig:clusterExp}(middle).
	Insert the copy of $W_k$ in $H'$. For every $1\leq i\leq d$ and $1\leq j\leq k-1$,
	insert a pipe $u_{p_j}u_{v_i}$ in $H'$ if an edge $p_jv_i'$ is present in $E(\skel(\mu_r))$.
	
	\textbf{(c) Changes in $G'$}.
	Delete all edges and the central vertex of $W_{k}$ from $G'$, which splits $C$ into $k-1$ components.
	Set $\varphi'(p_j)=u_{p_j}$ for $j\in\{1,\ldots,k-1\}$.
	By \ref{P4}, every cluster $u_{p_j}$ is adjacent to a single cluster, say $u_{v_i}$, outside of $W_k$.
	We modify $\varphi'$ so that it maps the vertices of the component of $C$ containing $p_j$ to $u_{v_i}$ with the exception of $p_j$, which is mapped to $u_{p_j}$. (Note that the cluster $u_C$ and all incident pipes are empty.)
	
	\noindent\textbf{Step 4: Local Planarity Test.}
	Let $H_u$ be the subgraph induced by the newly created clusters and pipes.
	Let $\widetilde{H_u}$ denote the graph obtained as the union of $H_u$
	and a star whose center is a new vertex (not in $V(H_u)$), and whose leaves
	are the clusters in $\partial\Delta_u$. Use a planarity testing algorithm to test
	whether $\widetilde{H_u}$ is planar.
	If $\widetilde{H_u}$ is not planar, report that $\phi$ is not a weak embedding and halt.
	Otherwise, find an embedding of $\widetilde{H_u}$ in which the center of the star is in the outer face.
	This defines a rotation system for $H_u$. The rotation system of $H'$ outside of $\Delta_u$ is
	inherited from $H$.
	
	\noindent\textbf{Step 5: Normalize.}
	Finally, apply \textsf{normalize} to each new cluster in $H'$.
	(This step subdivides edges so that \ref{P1} is satisfied as shown in Fig.~\ref{fig:clusterExp}(right).)
	
	\begin{figure}[h]
		\centering
		\def\svgwidth{\textwidth}
		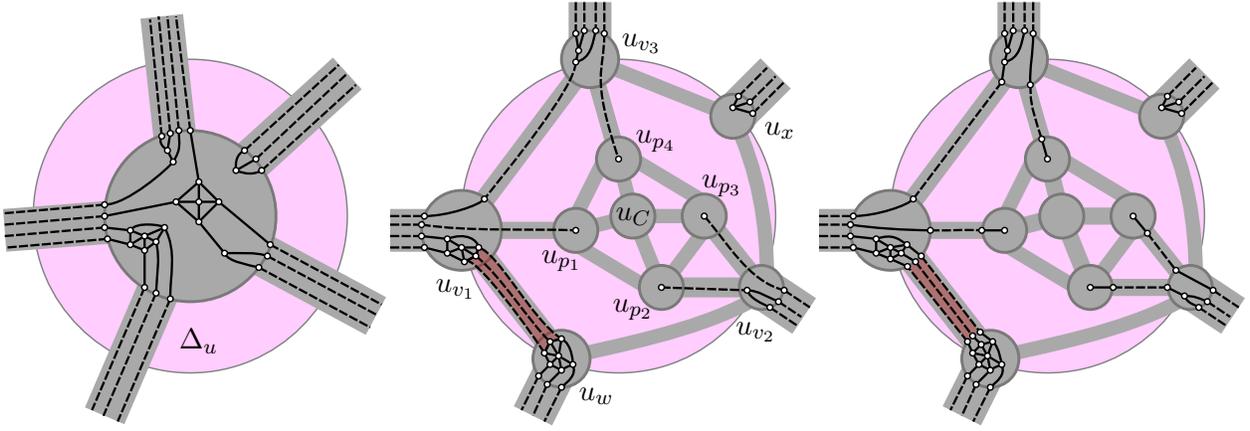
		\caption{\small
			Changes in a cluster caused by \textsf{clusterExpansion}.
			Left: input. Note that the embedding of $G_u$ is not known and the grey disks (rectangles) represent points (arcs). Middle: the instance after step 4. Right: output.	
			Red rectangles indicate triples of edges defining a thick edge.}
		\label{fig:clusterExp}
		\vspace{-.6\baselineskip}
	\end{figure}
	
	\begin{lemma}\label{lem:clusterExpansion}
		Given an instance $\phi:G\rightarrow H$ in simplified form containing a cluster $u$, $\textsf{clusterExpansion}(u)$ either reports that $\phi$ is not a weak embedding or produces an instance $\phi':G'\rightarrow H'$ in simplified form that is equivalent to $\varphi$ in $O(|E(G_u)|+\deg(u))$ time.
	\end{lemma}
	\begin{proof}
		All clusters of $H$ other than $u$ are clusters in $H'$, as well. By assumption, $\phi:G\rightarrow H$ is simplified form, hence these clusters satisfy \ref{P1}--\ref{P4} in $\phi'$. By construction, the new clusters satisfy \ref{P1}--\ref{P4} in $\phi'$, consequently $\phi':G'\rightarrow H'$ is in simplified form.
		
		Step~5 of the operation receives as input an instance $\phi^*:G^*\rightarrow H^*$ and returns an equivalent instance $\phi'$ by Lemma~\ref{lem:normalize}. It remains to show that $\phi^*$ and $\phi$ are equivalent, and to analyze the running time.
		
		First assume that $\phi$ is a weak embedding, and so there is an embedding $\psi_\phi:G\rightarrow\mathcal{H}$.
		We need to show that there exists an embedding $\psi_{\phi^*}:G^*\rightarrow \mathcal{H}^*$,
		and hence $\phi^*$ is a weak embedding.
		This can be done by performing steps~0--3 on the graphs $G$ and $H$, and the embedding $\psi_\phi$,
		which will produce $G^*$ and $H^*$, and an embedding $\psi_{\phi^*}:G^*\rightarrow \mathcal{H}^*$.
		By construction, every 2-cut in $H_u^*$ consists of a pair of clusters in $\partial\Delta_u$.
		Consequently, $H_u^*$ has a unique embedding with the given outer cycle.
		In particular, $\widetilde{H_u}^*$ is planar. The rotation of the new clusters of $H^*$ is uniquely defined and therefore must be consistent with any embedding of $G$, including $\psi_\phi$.

		Next, assume that $\phi^*$ is a weak embedding.
		Given an embedding $\psi_{\phi^*}:G^*\rightarrow \mathcal{H}^*$, we construct an embedding $\psi_\phi:G\rightarrow\mathcal{H}$ as follows.
		Let $H_u^*$ be the subgraph of $H^*$ induced by the clusters created by $\textsf{clusterExpansion}(u)$.
		Note that $H^*$ is a connected plane graph:
		the clusters created in step~0 are connected by a path (if $\deg(u)\leq 2$) or a cycle (if $\deg(u)\geq 3$);
		and any clusters created in step~3 (when $\deg(u)\geq 3$) are attached to this cycle.
		Since $H_u^*$ is a connected plane graph, we may assume that there is a topological disk containing
		only the pipes and clusters of $H_u^*$; let $D_u$ denote such a topological disk.
		
		Let $G_u^*$ be the subgraph of $G^*$ mapped to $H_u^*$. We show that
		steps 0--3 of the operation can be reversed without introducing crossings.
		For each component $C$ of $G_u$ with $\text{pipe-deg}(C)\ge 3$,
		embed a wheel $W_k$ in the disk $D_{u_C}$ around the cluster $u_C$,
		and connect its external vertices to the vertices $p_i$, $i=1,\ldots ,k-1$,
		Since the pipes incident to $u_C$ are empty, and each $p_i$ is a unique vertex in
		its cluster, this can be done without crossings. Now, every component $C$ of
		$G_u$ corresponds to a component $C^*$ of $G_u^*$, and by \ref{P1}
		every terminal vertex in $C$ corresponds to terminal in $C^*$.
		
		If we delete a component $C^*$ from $G^*$, there will be a face $F$ of $D_u$
		(a component of $D_u\setminus \psi_{\phi^*}(G^*)$) that contains all terminals of $C$ on its boundary. Denote by $\pi_C$ the ccw cyclic order in which these terminals appear in the facial walk of $F$.
		If $C$ admits an embedding in which the terminals appear in the outer face in the same order as $\pi_C$, we can embed $C$ in $F$ on the given terminals.	We show that $C$ admits such an embedding by proving that the SPQR trees of $\overline{C}$ and $\overline{C}^*$ impose the same constraints on the cyclic order of terminals.
		Subdividing edges do not change these constraints.
		Step~1 does not change $C$.
		If step~2(b2) adds an edge in the skeleton of an S node, the possible combinatorial embeddings remain the same.
		Step~2(b3) maintains the rotation of $v'$ and $w'$.
		If $\text{pipe-deg}(C)\ge3$, $C$ and $C^*$ are identical apart from subdivided edges.
		Because all steps maintain the same constraints on the rotation system of the terminals,
		we can construct $\psi_\phi:G\rightarrow\mathcal{H}$ by incrementally replacing the embedding
		of $C^*$ by an embedding of $C$ in $D_u$ for every component $C$ of $G_u$.
		By Lemma~\ref{lem:comb-reconstruction}, this is possible without introducing crossings since no two components of $G_u$ cross and every component is planar since $\phi$ is in simplified form.
		
		Finally, we show that $\textsf{clusterExpansion}(u)$ runs in $O(|E(G_u)|+\deg(u))$ time.
		Step~0 takes $O(\deg(u))$ time.
		Steps~1--3 are local operations that take $O(|E(C)|)$ time for each  component $C$ of $G_u$.
		Step~4 takes $O(|E(G_u)|+\deg(u))$ time since each component $C$ inserts at most $O(|E(C)|)$ pipes in $H'$, where $|E(G_u)|=\sum_C |E(C)|$; the graph $\widetilde{H_u}$ has $\deg(u)$ more edges than $H_u$, and planarity testing takes linear time in the number of edges~\cite{HoTa74_planarity}.
		Step~5 takes $O(|E(G_u)|)$ time by Lemma~\ref{lem:normalize} and \ref{P1}.
	\end{proof}

	\subsection{Pipe Expansion}
	A cluster $u\in V(H)$ is called a \textbf{base of} an incident pipe $uv\in E(H)$ if every component of $G_u$ is incident to:
	\textbf{(i)}  at least one pipe-edge in $uv$; and
	\textbf{(ii)}  at most one pipe-edge or one thick pipe-edge (i.e., a triple of pipe-edges), other than those in pipe $uv$.
	We call a pipe $uv$ \textbf{safe} if both of its endpoints are bases of $uv$; otherwise it is \textbf{unsafe}.
	For a simplified instance $\varphi:G\rightarrow H$ and a safe pipe $uv\in E(H)$,
	operation \textsf{pipeExpansion}$(uv)$ consists of the following steps:
	
	\paragraph{\textsf{pipeExpansion}$(uv)$.}
	Input: an instance $\varphi:G\rightarrow H$ in simplified form and a safe pipe $uv\in E(H)$.
	The operation either reports that $\varphi$ is not a weak embedding or returns an instance $\varphi':G'\rightarrow H'$.
	First, produce an instance $\varphi^*:G\rightarrow H^*$ by contracting the pipe $uv$ into the new cluster
	$\cluster{uv}\in V(H^*)$ while mapping the vertices of $G_u$ and $G_v$ to $\cluster{uv}$, see Fig.~\ref{fig:PipeExp}(top);
	and then apply \textsf{simplify} and \textsf{clusterExpansion} to $\cluster{uv}$.
	
	\begin{figure}[h]
		\centering
		\def\svgwidth{0.7\textwidth}
		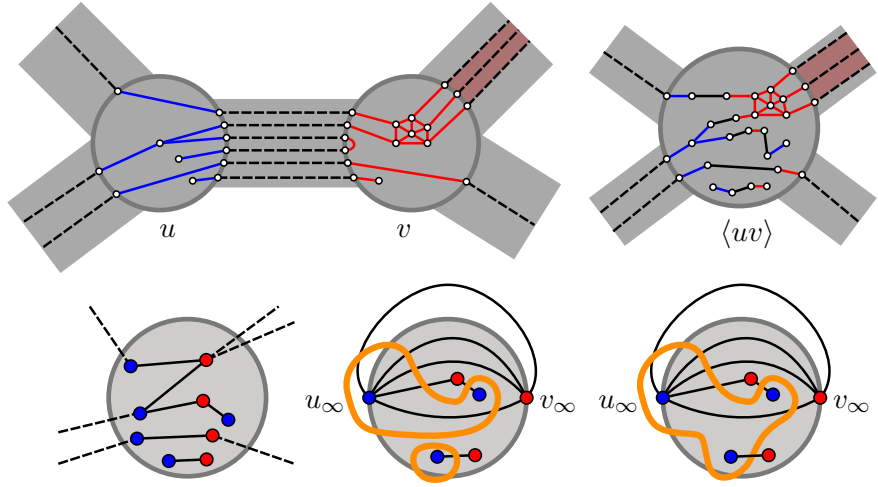
		\caption{\small
			Pipe Expansion. A safe pipe $uv$ (top left).
			The cluster $\cluster{uv}$ obtained after contraction of $uv$ (top right).
			The result of contracting the components in $G_u$ and $G_v$ (bottom left).
			The subsequent contraction of all components incident to $\partial_u D_{\cluster{uv}}$ and $\partial_v D_{\cluster{uv}}$ (all loops are already deleted),
			resp., and a Jordan curve that crosses every edge of the resulting bipartite plane multigraph  for each component (bottom center). All Jordan curves can be combined into one (bottom right).}
		\label{fig:PipeExp}
	\end{figure}
	
	We use the following folklore result in the proof of correctness of operation $\textsf{pipeExpansion}(uv)$.
	This result is obtained by an Euler tour algorithm on the dual graph of a plane bipartite multigraph.
	\begin{theorem}[Belyi~\cite{B83_self}]
		\label{thm:belyi}
		For every embedded connected bipartite multigraph $G^*$, there exists a Jordan curve
		that crosses every edge of $G^*$ precisely once.
		Such a curve can be computed in $O(|E(G^*)|)$ time.
	\end{theorem}
	
	\begin{lemma}\label{lem:pipeExpansion}
		Given an instance $\varphi:G\rightarrow H$ and a safe pipe $uv\in E(H)$,
		$\textsf{pipeExpansion}(uv)$ either reports that $\varphi$ is not a weak embedding or produces an equivalent instance $\varphi':G'\rightarrow H'$.
	\end{lemma}
	\begin{proof}
		Let $\varphi:G\rightarrow H$ be an instance in simplified form, and let $uv$ be a safe pipe.
		Recall that $\textsf{pipeExpansion}(uv)$ starts with producing an instance $\varphi^*:G\rightarrow H^*$ by contracting the pipe $uv$ into the new cluster $\cluster{uv}\in V(H^*)$ while mapping the vertices of $G_u$ and $G_v$ to $\cluster{uv}$.
		It is enough to prove that $\varphi$ and $\varphi^*$ are equivalent, the rest of the proof follows from Lemmas~\ref{lem:normalize}, \ref{lem:simplify}, and \ref{lem:clusterExpansion}.
		
		One direction of the equivalence proof is trivial: Given an embedding $\psi_\phi:G\rightarrow\mathcal{H}$, we can obtain an embedding $\psi_{\phi^*}:G\rightarrow\mathcal{H}^*$ by defining $D_{\cluster{uv}}$ as a topological disk containing only $D_u$, $D_v$, and $R_{uv}$.
		
		For the other direction, assume that we are given an embedding $\psi_{\phi^*}:G\rightarrow\mathcal{H}^*$. We need to show that there exists an embedding $\psi_\phi:G\rightarrow\mathcal{H}$.
		We shall apply Theorem~\ref{thm:belyi} after contracting certain subgraphs of $G_u$ and $G_v$ (as described below). If a component contains cycles, then the contraction of its embedding creates a bouquet of loops. We study the cycles formed by $G_u$, $G_v$, and thick edges incident to $G_u$ or $G_v$ to ensure that no other component is embedded in the interior of such cycles.
		
		\smallskip\noindent\textbf{Components of $G_{\cluster{uv}}$ of pipe-degree 0.}
		Note that the terminals corresponding to the pipes incident to $u$ and $v$ lie in two disjoint arcs of $\partial D_{\cluster{uv}}$,
		which we denote by $\partial_u D_{\cluster{uv}}$ and $\partial_v D_{\cluster{uv}}$, respectively. The components of graph $G_{\cluster{uv}}$
		with positive pipe-degree are incident to terminals in $\partial_u D_{\cluster{uv}}$ or $\partial_v D_{\cluster{uv}}$ (possibly both).
		The components of pipe-degree 0 can be relocated to any face of the embedding of all other components.
		Without loss of generality, we may assume that all components of pipe-degree 0 lie in a common face
		incident to both $G_u$ and $G_v$ in $\psi_{\phi^*}$.
		
		\smallskip\noindent\textbf{Cycles induced by $G_u$ and $G_v$, and by thick edges.}  Notice that \ref{P3} and \ref{P4} imply that every maximal biconnected component in $G_u$ and $G_v$ is a wheel.
		Since each wheel is 3-connected, the circular order of their external vertices is determined by the embedding $\psi_{\phi^*}$.
		We may assume that in the embedding no cycle of a wheel subgraph encloses any vertex other than the center of the wheel. Indeed, suppose a 3-cycle $(p_1,p_2,p_c)$ of a wheel encloses some vertex, where $p_c$ is the center of the wheel, and $p_1$ and $p_2$ are two consecutive external vertices. We can modify the embedding of the edge $p_1p_2$ in $\psi_{\phi^*}$ so that it follows closely the path $(p_1,p_c,p_2)$ and the 3-cycle does not contain any vertex; see Fig.~\ref{fig:changingThick}(a).
		
		Consider a thick edge $\theta$ in $G$ between a wheel in $G_u$ (or $G_v$) and a wheel in $G_w$ for some adjacent cluster $w\notin\{u,v\}$.  Recall that a thick edge gives rise to three paths, say $P_1=(p_1,t_1,t_2,p_2)$, $P_2=(p_3,t_3,t_4,p_4)$, and $P_3=(p_5,t_5,t_5,p_6)$, where $(p_1,p_3,p_5)$ and $(p_2,p_4,p_6)$ are consecutive external vertices of the two wheels, resp., and $p_i$ is the unique vertex in a cluster adjacent to terminal $t_i$ for $i\in\{1,\ldots,6\}$; cf.~\ref{P1}.
		If we suppress the terminals, then the two wheels incident to the thick edge would be in the same maximal 3-connected component of $G$ and, therefore, their relative embedding is fixed within $D_u\cup R_{uw}\cup D_w$. We may assume that no vertex is enclosed by a cycle induced by the vertices of the thick edge (i.e., by any pair of paths from $P_1$, $P_2$, and $P_3$) in $D_u\cup R_{uw}\cup D_w$. Indeed, we can modify the embedding of the path $P_1$ and $P_3$ in $\psi_{\phi^*}$ so that they closely follow the path $p_1p_3\cup P_2\cup p_4p_2$ and $p_5p_3\cup P_2\cup p_4p_6$, respectively. By \ref{P4}, such modification of the embedding is always possible without introducing crossings; see Fig.~\ref{fig:changingThick}(b). We conclude that a cycle induced by the thick edge $\theta$ does not enclose any vertices of $G$.
		
		\begin{figure}[htbp]
			\centering
			\def\svgwidth{.6\textwidth}
\begingroup%
  \makeatletter%
  \providecommand\color[2][]{%
    \errmessage{(Inkscape) Color is used for the text in Inkscape, but the package 'color.sty' is not loaded}%
    \renewcommand\color[2][]{}%
  }%
  \providecommand\transparent[1]{%
    \errmessage{(Inkscape) Transparency is used (non-zero) for the text in Inkscape, but the package 'transparent.sty' is not loaded}%
    \renewcommand\transparent[1]{}%
  }%
  \providecommand\rotatebox[2]{#2}%
  \newcommand*\fsize{\dimexpr\f@size pt\relax}%
  \newcommand*\lineheight[1]{\fontsize{\fsize}{#1\fsize}\selectfont}%
  \ifx\svgwidth\undefined%
    \setlength{\unitlength}{177.14601517bp}%
    \ifx\svgscale\undefined%
      \relax%
    \else%
      \setlength{\unitlength}{\unitlength * \real{\svgscale}}%
    \fi%
  \else%
    \setlength{\unitlength}{\svgwidth}%
  \fi%
  \global\let\svgwidth\undefined%
  \global\let\svgscale\undefined%
  \makeatother%
  \begin{picture}(1,0.62401278)%
    \lineheight{1}%
    \setlength\tabcolsep{0pt}%
    \put(0,0){\includegraphics[width=\unitlength,page=1]{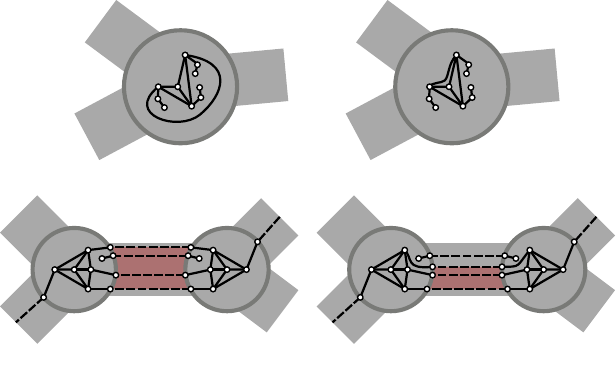}}%
    \put(0.46055916,0.33348051){\color[rgb]{0,0,0}\makebox(0,0)[lt]{\lineheight{1.25}\smash{\begin{tabular}[t]{l}(a)\end{tabular}}}}%
    \put(0.45869179,0.00917876){\color[rgb]{0,0,0}\makebox(0,0)[lt]{\lineheight{1.25}\smash{\begin{tabular}[t]{l}(b)\end{tabular}}}}%
  \end{picture}%
\endgroup%

			\caption{\small
				Changing the embedding so that: (a) no cycle encloses vertices that are not the center of a wheel; and (b) no induced cycle of a thick edge encloses a vertex.}
			\label{fig:changingThick}
		\end{figure}
		
		\smallskip\noindent\textbf{Separating $G_u$ and $G_v$.}
		We next show that there exits a closed Jordan curve that separates the embeddings of $G_u$ and $G_v$, and crosses every edge between $G_u$ and $G_v$ precisely once. In order to use Theorem~\ref{thm:belyi}, we reduce $G_{\cluster{uv}}$ and its embedding to a bipartite multigraph in two steps.
		
		\noindent\textbf{(1)}
		Contract each component $C$ of $G_u$ (resp., $G_v$) to a single vertex $w_C$,
		see Fig.~\ref{fig:PipeExp}(bottom-left).
		This results in a bouquet of loops at $w_C$ that corresponds to faces bounded by $C$. As argued above, a cycle through the external vertices
		of a wheel in $G_u$ or $G_v$ encloses only its center.
		Hence, none of the loops at $w_C$ encloses any other vertex of $G$, and they can be discarded.
		
		\noindent\textbf{(2)}
		As $uv$ is safe, every component $C$ of $G_u$ and $G_v$ is incident to either at most one terminal in $\partial D_{\cluster{uv}}$ or precisely three terminals corresponding to a thick edge. Contract the arc $\partial_u D_{\cluster{uv}}$ (resp., $\partial_v D_{\cluster{uv}}$) and for every component $C$ of $G_u$ (resp.,  $G_v$)  all edges between $w_C$ and terminals into a new vertex $u_{\infty}\in \partial D_{\cluster{uv}}$ (resp., $v_{\infty}\in\partial D_{\cluster{uv}}$), see Fig.~\ref{fig:PipeExp}(bottom-center);
		and insert an edge $u_\infty v_\infty$ in the exterior of $D_{\cluster{uv}}$. This results in a bouquet of loops at $u_{\infty}$ (resp., $v_{\infty}$), corresponding to thick edges. As argued above, these loops do not enclose any vertices, and can be eliminated.
		The subgraph of all components of pipe-degree 1 or higher has been transformed into a connected bipartite plane multigraph, and each component of pipe-degree 0 is also transformed into a connected bipartite plane multigraph.
		
		We apply Theorem~\ref{thm:belyi} for each of these plane multigraphs independently.
		By our assumption, the components $C$, $\text{pipe-deg}(C)=0$, lie in a common face.
		Consequently, we can combine their Jordan curves with the Jordan curve of all remaining components
		into a single Jordan curve that crosses every edge between $G_u$ and $G_v$ precisely once, see Fig.~\ref{fig:PipeExp}(bottom-right).
		After reversing the contractions and loop deletions described above, we obtain a Jordan curve $\gamma$
		that separates $G_u$ and $G_v$, and crosses every edge between $G_u$ and $G_v$ precisely once.
		
		Since $\gamma$ crosses the edge $u_\infty v_\infty$ precisely once, and $u_\infty v_\infty$ is homotopic to either subarcs of $\partial D_{\cluster{uv}}$ between $u_\infty$ and $v_\infty$, we can assume that $\gamma$ crosses
		each of these arcs precisely once. Consequently, the Jordan arc $\gamma'=\gamma\cap D_{\cluster{uv}}$
		partitions the disk $D_{\cluster{uv}}$ into two disks, $D_u$ and $D_v$, adjacent to $\partial_u D_{\cluster{uv}}$ and $\partial_v D_{\cluster{uv}}$, respectively, such that every edge between $G_u$ and $G_v$ crosses $\gamma'$ precisely once and all other edges of $G_{\cluster{uv}}$ lie entirely in either $D_u$ or $D_v$.
		This yields the required embedding $\psi_\phi:G\rightarrow\mathcal{H}$.
	\end{proof}
	
	\section{Algorithm and Runtime Analysis}
	\label{sec:alg}
	
	In this section, we present our algorithm for recognizing weak embeddings.
	We describe the algorithm and prove that it recognizes weak embeddings in Section~\ref{ssec:alg1}.
	A na\"ive implementation would take $O(m^2)$ time as explained below; we describe how to implement it in $O(m\log m)$ time using additional data structures in Section~\ref{ssec:alg2}.
	
	\subsection{Main Algorithm}
	\label{ssec:alg1}
	
	We are given a piecewise linear simplicial map $\varphi:G\rightarrow H$, where $G$ is a graph and $H$ is an embedded graph in an orientable surface.
	We introduce some terminology for an instance $\varphi:H\rightarrow G$.
	\begin{itemize}
		\item A component $C$ of $G_u$, in a cluster $u\in V(H)$, is \textbf{stable} if $\text{pipe-deg}(C)=2$ and,
		in each of the two pipes, $C$ is incident to exactly one edge or exactly one thick edge.
		Otherwise $C$ is \textbf{unstable}.
		\item Recall that a pipe $uv\in E(H)$ is \textsl{safe} if both $u$ and $v$ are bases of $uv$; otherwise it is \textsl{unsafe}.
		\item A cluster $u\in V(H)$ (resp., a pipe $uv\in E(H)$) is \textbf{empty} if $\varphi^{-1}[u]=\emptyset$ (resp., $\varphi^{-1}[uv]=\emptyset$). Otherwise $u$ is \textbf{nonempty}.
		\item A safe pipe $uv\in E(H)$ is \textbf{useless} if $uv$ is empty or if $u$ and $v$ are each incident to exactly  two nonempty pipes and every component of $G_u$ and $G_v$ is stable. The safe pipe $uv$ is \textbf{useful} otherwise
		(i.e., if $uv$ is nonempty; and $u$ or $v$ is incident to at least 3 nonempty pipes, or $G_u$ or $G_v$ has an unstable component).
		\item A \textbf{generalized cycle} is a cycle in which each node is a component of $G_u$ for some cluster $u$, and each edge is either an edge or a thick edge between wheels.
	\end{itemize}
	
	Notice that if a pipe $uv$ is useless, then \textsf{pipeExpansion}$(uv)$ returns a combinatorially equivalent instance
	(i.e., there would be no progress).
	As we show below (Lemma~\ref{lem:all-safe}), our algorithm reduces $G$ to a collection of generalized cycles. (In particular, if we contract every wheel to a single vertex, then a generalized cycle reduces to a cycle, possibly with multiple edges coming from thick edges.)
	
	\paragraph{Data Structures.}
	The graph $G$ is stored using adjacency lists.
	We store the combinatorial embedding of $H$ using the rotation system of $H$ (ccw order of pipes incident to each $u\in V(H)$).
	The mapping $\phi$ is encoded by its restriction to $V(G)$: that is, by the images $\varphi(p)$ for all $p\in V(G)$.
	For each cluster $u\in V(H)$, we store the set of components of $G_u$, each stored as the ID of a representative edge
	and the pipe-degree of the component. Every pipe $uv\in E(H)$ has two Boolean variables to indicate whether the respective
	endpoints are its bases; and an additional Boolean variable indicates whether $uv$ is useful.
	These data structures are maintained dynamically as the algorithm modifies $G$, $H$, and $\varphi$.
	
	\paragraph{Algorithm$(\phi)$.}
	Input: an instance $\varphi:G\rightarrow H$ in simplified form.
	
	\smallskip\noindent\textbf{Phase 1.}
	Apply \textsf{clusterExpansion} to each $u\in V(H)$.
	Denote the resulting instance by $\varphi':G'\rightarrow H'$.
	Build the data structures described above for $\phi':G'\rightarrow H'$.
	
	\smallskip\noindent\textbf{Phase 2.}
	While there is a useful pipe in $H'$, let $uv\in E(H')$ be an arbitrary useful pipe
	and apply \textsf{pipeExpansion}$(uv)$. 
	
	\smallskip\noindent\textbf{Phase 3.}
	If any component of $G'$ that contains a wheel is nonplanar, then report that $\phi$ is not a weak embedding and halt.
	Otherwise, in each cluster contract every wheel component to a single vertex, and turn every thick edge into a single
	edge by removing multiple edges. Denote the resulting instance by $\varphi'':G''\rightarrow H''$.
	If any component $C$ of $G''$ is a cycle with $k$ vertices but $\varphi''(C)$ is not a cycle with $k$ clusters in $H''$, then report that $\phi$ is not a weak embedding, else report that $\varphi$ is a weak embedding. This completes the algorithm.
	
	\subsection{Analysis of Algorithm}
	We show that the Algorithm recognizes whether the input $\phi:G\rightarrow H$ is a weak embedding.
	The running time analysis follows in Section~\ref{ssec:alg2}.
	
	\paragraph{Termination.}
	Since both Phase~1 and Phase~3 consist of for-loops only, it is enough to show that the while loop in Phase~2 terminates.
	We define a nonnegative potential function $\Phi_1(\phi)$ for an instance $\phi:G\rightarrow H$.
	For a pipe $uv\in E(H)$, let $\sigma(uv)$ be the number of pipe-edges in $\phi^{-1}[uv]$ minus twice the number of thick edges (so that each thick edge is counted as one).
	Let $s(u)$ be the number of stable components of $G_u$.
	We define the following quantities for an instance $\phi:G\rightarrow H$:
	\begin{equation}
		\begin{split}
			N_\sigma(\phi)= \sum_{uv\in E(H)} \sigma(uv),\\
			N_s(\phi)=\sum_{u\in V(H)}s(u).
		\end{split}
	\end{equation}

	Let $Q(\phi)$ be the number of nonempty pipes of an instance $\phi$.
	We can now define the potential function
	\begin{equation}\label{Phi:1}
		\Phi_1(\phi)= 4(N_\sigma(\phi)-Q(\phi))+N_s(\phi).
	\end{equation}
	Note that $\Phi_1(\phi)$ is a nonnegative integer since $N_\sigma(\phi)\ge Q(\phi)$.

	The following lemma describes the effect of one iteration of the while loop in Phase~2 on $\Phi_1$ and $N_\sigma-Q$.
	
	\begin{lemma}\label{lem:termination}
		Assume we invoke \textsf{pipeExpansion}$(uv)$ for a useful pipe $uv$ in an instance $\varphi$, obtaining instance $\phi'$. Then,
		\begin{itemize}
			\item $\Phi_1(\varphi)>\Phi_1(\varphi')$, and
			\item $N_\sigma(\varphi)-Q(\varphi)\ge N_\sigma(\varphi')-Q(\varphi')$.
		\end{itemize}
	\end{lemma}
	\begin{proof}
		Let $\cluster{uv}$ be the cluster obtained after the first step of \textsf{pipeExpansion}$(uv)$ (before applying \textsf{simplify}), and let $C$ be a component of $G_{\cluster{uv}}$. Let $\sigma_C(uv)$ be the number of pipe-edges of $C$ in $\phi^{-1}[uv]$ minus twice the number of thick edges of $C$ in $\phi^{-1}[uv]$ before \textsf{pipeExpansion}$(uv)$.
		Let $s_C$ be the number of stable components of $G_u$ and $G_v$ contained in $C$ before \textsf{pipeExpansion}$(uv)$. Then
		$$\sigma(uv)=\sum_{C} \sigma_C(uv)
		\hspace{1cm}\text{and}\hspace{1cm}
		s(u)+s(v)=\sum_{C} s_C,$$
		where the summation is over all components $C$ of $G_{\cluster{uv}}$.
		We distinguish cases based on the pipe-degree of $C$ and analyse their contribution to $N_\sigma-Q$ and $N_s$ before and after \textsf{pipeExpansion}$(uv)$.
		The case when $\text{pipe-deg}(C)=0$ is trivial since $C$ disappears and therefore its contribution to $N_\sigma-Q$ and $N_s$ decreases.
		
		Assume $\text{pipe-deg}(C)=1$. Then $\sigma_C(uv)\geq 1$ by the definition of safe pipes.
		Step~1 of \textsf{clusterExpansion} (called inside \textsf{pipeExpansion}$(uv)$) creates component $C'$ of pipe-degree~1 in some cluster on the boundary of $\Delta_{\cluster{uv}}$ and no new pipe-edge in $\Delta_{\cluster{uv}}$.
		Consequently, \textsf{pipeExpansion} decreases the contribution of $C$ to $N_\sigma$ and $N_s$ to 0.
		The contribution of $C$ to $N_\sigma-Q$ does not decrease (i.e., remains constant 0) only if
		$\sigma_C(uv)=1$ and $C$ is the only component in $G_{\cluster{uv}}$. In this case, however,
		$C$ contains precisely one component in each of $G_u$ and $G_v$, of pipe-degree 1 and 2.
		The component of pipe-degree 2 is stable since $uv$ is safe and $\sigma_C(uv)=1$,
		and so the number of stable components decreases by one,
		hence \textsf{pipeExpansion} decreases the contribution of $C$ to $\Phi_1$.
		
		Assume $\text{pipe-deg}(C)=2$ (see Fig.~\ref{fig:flip}).
		Then Step~2 of \textsf{clusterExpansion} (called inside \textsf{pipeExpansion}$(uv)$) creates two stable components, connected by either a single edge or a thick edge in a pipe in $\Delta_{\cluster{uv}}$.
		Consequently, \textsf{pipeExpansion}$(uv)$ changes the contribution of $C$ to $N_\sigma$ from at least 1 to precisely 1.
		The number of stable components increases if $s_C\in \{0,1\}$, 
		in which case at least one component of $G_u$ and $G_v$ in $C$ is unstable.
		Since $uv$ is useful, such a component is incident to at least two single or thick edges in $uv$,
		and so $\sigma_C(uv)\geq 2$. In this case, the contribution of $C$ to $N_\sigma$ decreases by at least one, while the number of stable components increases by at most two, hence $\Phi_1$ strictly decreases.
		The contribution of $C$ to $\Phi_1$ is unchanged only if $C$ contains precisely one component in each of $G_u$ and $G_v$, where $\sigma_C(uv)=1$ and $s_C=2$ (that is, both are stable components).

		Assume $\text{pipe-deg}(C)\ge3$.
		Let $W_k$, $k\ge 4$, be the wheel created by \textsf{clusterExpansion}$(\cluster{uv})$.
		Since $uv$ is a safe pipe, every component of $G_u$ or $G_v$ contained in $C$ is
		incident to at most one single or thick pipe-edge outside of $\varphi^{-1}[uv]$.
		Hence, $C$ contains at least $k-2$ edges between these components of $G_u$ and $G_v$
		(which are pipe-edges in $uv$), and its contribution to $N_\sigma$ is at least $k-2$.
		Step~3 of \textsf{clusterExpansion}$(uv)$ (called inside \textsf{pipeExpansion}$(uv)$)
		creates $k-1$ components in distinct clusters of $\Delta_{\cluster{uv}}$.
		Hence, the contribution of $C$ to $N_s$ increases by at most $k-1$.
		Step~3 also creates $k-1$ pipe-edges, each of which is the only pipe-edges in its pipe,
		hence they contribute zero to $N_\sigma-Q$.
		The difference between the new and the old value of $N_\sigma-Q$ due to $C$ is at most $-(k-2-1)=-k+3$
		(where the term $-1$ accounts for the case that $C$ is the only component of $G_{\cluster{uv}}$).
		Hence, we obtain the following upper bound on the difference between the new and the old contribution of $C$ to $\Phi_1$:  $-4(k-3)+(k-1) = 11-3k \le 11-3\cdot 4 = -1 < 0$, where the $(k-1)$-term corresponds to the increase in $N_s$.
		Consequently the contribution of $C$ to both, $\Phi_1$ and
		$N_\sigma-Q$, decreases.
		
		Summing over the contributions of all components of $G_{\cluster{uv}}$, we have shown that
		neither $\Phi_1$ nor $N_\sigma-Q$ increases. 
		It remains to show that the $\Phi_1$ decreases.
		By the previous case analysis, the contribution of a component $C$ to neither $\Phi_1$ nor $N_\sigma-Q$ increases.
		The contribution of $C$ to $\Phi_1$ is unchanged when $C$ contains exactly one component in each $G_u$ and $G_v$, and both are stable, and hence the contribution of $C$ to $N_s$ is remains 2.
		For contradiction, assume that $\Phi_1$ and $N_\sigma-Q$ are both unchanged.
		Then all components of $G_{\cluster{uv}}$ are stable, and each contains precisely one (stable) component in $G_u$ and $G_v$, respectively. Then $N_s$ and $N_\sigma$ remain unchanged.
		Therefore, $Q$ must also remain unchanged, which implies that all components of $G_{\cluster{uv}}$ are adjacent to the same pair of pipes. Then, $uv$ is useless, a contradiction.
	\end{proof}
	\begin{corollary}\label{cor:termination}
		The algorithm executes at most $15m$ iterations of Phase~2 and, therefore, terminates;
		and it creates at most $12m$ stable components.
	\end{corollary}
	\begin{proof}
		%
		First we show that at the end of Phase~1, there are at most $3m$ pipe-edges.
		We charge the creation of $\text{pipe-deg}(C)$ pipe-edges in new pipes in Step 3 of cluster expansion to $\text{pipe-deg}(C)$ original pipe-edges incident to $C$.
		Similarly we charge the creation of a (thick) pipe-edge created in Step 2 to 1 or 3 (in the case we created a thick edge) original pipe-edges incident to $C$.
		In the latter case, $C$ must be incident to at least 4 pipe-edges.
		Clearly, each original edge receives the charge of at most 2, one from each of its incident components. Hence, in total we introduced at most $2m$ new pipe-edges.
		
		By definition $N_\sigma-Q$ is a nonnegative integer, and its initial value does not exceed the number of pipe-edges.
		Therefore, at the end of Phase~1, we have $0\leq N_\sigma-Q\leq 3m$.
		Since each pipe-edge can be incident to at most two stable components and each stable component must be incident to two pipe-edges, we have $0\leq N_s\leq 3m$.
		Hence, $0\leq \Phi_1\leq 15m$.
		
		By Lemma~\ref{lem:termination}, $\Phi_1$ strictly decreases in each iteration of the while loop in Phase~2.
		Therefore Phase~2 has at most $15m$ iterations.
		Since both Phase~1 and Phase~3 consist of for-loops only, the algorithm terminates.
		
		By Lemma~\ref{lem:termination}, $\Phi_1$ strictly decreases and $N_\sigma-Q$ does not increase in each iteration of the while loop in Phase~2.
		Consequently, we can charge the creation of a stable component to the decrease of one fourth of a unit of $N_\sigma-Q$.
		Hence, Phase~2 creates at most $12m$ stable components.
	\end{proof}
	
	\paragraph{Correctness.}
	We show next that Phase~2 reduces $G'$ to a collection of generalized
	cycles (see Lemma~\ref{lem:all-safe} below). First, we need a few observations.
	
	\begin{lemma}\label{lem:inv}
		Every cluster $u'$ created by an operation \textsf{clusterExpansion}$(u)$ satisfies the following:
		\begin{enumerate}[label=(B\arabic*)]
			\item \label{b1} $u'$ is either empty or the base for some nonempty pipe.
			\item \label{b2} If $u'$ lies in the interior of the disk $\Delta_u$,
			then $u'$ is either empty or the base of a unique nonempty pipe $u'v'$,
			where $v'$ in on the boundary of $\Delta_u$.
			\item \label{b3} If $u'$ lies on the boundary of the disk $\Delta_u$,
			then $u'$ is either empty, or the base of at least one and at most two pipes, exactly one of which is outside of $\Delta_u$. (In particular, the pipe outside of $\Delta_u$ is unique, this property will be crucial.)
			\item \label{b4} If $u'$ is nonempty,
			then $u'$ is incident to at most three empty pipes.
		\end{enumerate}
	\end{lemma}
	\begin{proof}
		Operation \textsf{clusterExpansion}$(u)$ creates a cluster $u_v$ on the boundary of the disk $\Delta_u$
		for every pipe $uv$ incident to $u$. By construction, $u_v$ is either empty (if the pipe $uv$ was empty),
		or a base for the pipe $vu_v$, which lies outside of $\Delta_u$.
		When $u_v$ is nonempty, the only incident pipes that could be empty are created in Step 0 (either two pipes of the cycle $C_u$ or a single pipe if $\deg(u)=2$).
		For each component $C$ of $G_u$ with $\text{pipe-deg}(C)\geq 3$,
		\textsf{clusterExpansion}$(u)$ creates a wheel where
		the center cluster is empty and every external cluster is incident to
		exactly one nonempty pipe (to a cluster on the boundary of $\Delta_u$),
		consequently, it is a base for that pipe.
		The external clusters of the wheel are incident to exactly three empty pipes by \ref{P4}. By construction, clusters in $\partial \Delta_u$ are incident to at most two empty pipes.
	\end{proof}
	
	\begin{corollary}\label{cor:inv}
		In every step of Phase~2, every cluster $u\in V(H')$ is either empty
		or the base for at least one nonempty pipe.
	\end{corollary}
	\begin{proof}
		Phase~1 of the algorithm performs \textsf{clusterExpansion}$(u)$ for every cluster of the input independently.
		At the end of Phase~1 (i.e., beginning of Phase~2), each cluster in $H'$ has been created by a
		\textsf{clusterExpansion} operation, and \ref{b1} follows from Lemma~\ref{lem:inv}.
		
		Subsequent steps of Phase~2  successively apply \textsf{pipeExpansion} operations
		thereby creating new clusters. By Lemma~\ref{lem:inv},
		property \ref{b1} is established for all new clusters,
		and it continues to hold for existing clusters.
	\end{proof}
	
	\begin{lemma}\label{lem:emptysafe}
		In every step of Phase~2, if $H'$ contains a nonempty unsafe pipe, then it also contains a useful pipe.
	\end{lemma}
	\begin{proof}
		Let $u_0u_1$ be a nonempty unsafe pipe in $H'$. Without loss of generality, assume that $u_1$ is not a base for $u_0u_1$.
		We iteratively define a simple path $(u_0,u_1,\ldots , u_\ell)$ for some $\ell\in \mathbb{N}$ as follows. Assume that $i\geq 1$,
		vertices $u_0,\ldots , u_i$ have been defined, and $u_{i-1}u_i$ is a nonempty unsafe pipe in $H'$ for which $u_i$ is not a base. By \ref{b1}, $u_i$ is a base for some nonempty pipe $u_iw$, where $w\neq u_{i-1}$. Put $u_{i+1}=w$.
		Since $H$ is simple, this iterative process terminates when either $u_iu_{i+1}$ is a nonempty safe pipe, or $u_{i+1}=u_k$ for some $0\leq k< i-1$.
		
		\smallskip\noindent\textbf{Case~1: The iterative process finds a nonempty safe pipe $u_iu_{i+1}$.}
		We claim that $u_iu_{i+1}$ is useful. Suppose, to the contrary, that $u_iu_{i+1}$ is useless. Then $u_i$ is incident to two nonempty pipes (which are necessarily $u_{i-1}u_i$ and $u_iu_{i+1}$), and every component in $G_{u_i}$ has pipe-degree 2. Consequently, $u_i$ is a base for both $u_{i-1}u_i$ and $u_iu_{i+1}$. This contradicts our assumption that $u_i$ is not a base for $u_{i-1}u_i$; and proves the claim.
		
		\smallskip\noindent\textbf{Case~2: The iterative process finds a cycle $U=(u_k,u_{k+1},\ldots , u_{\ell})$ where $u_{\ell+1}=u_k$ and for every $i=k,\ldots , \ell$, $u_i$ is a base for $u_iu_{i+1}$ but not a base for $u_iu_{i-1}$.}
		We show that this case does not occur. All clusters in the cycle have been created by \textsf{clusterExpansion} operations (in Phase~1 or 2).
		We claim that not all clusters in $U$ are created by the same \textsf{clusterExpansion} operation.
		For otherwise, by~\ref{b2} and~\ref{b3} $u_i$ and $u_{i+1}$, for some $i$, are on the boundary of an expansion disk $\Delta_u$ and $u_{i+1}$ is not the base of $u_iu_{i+1}$ by the construction of $U$. Furthermore, $u_iu_{i+1}$ is nonempty by the definition of a base. By~\ref{b3}, $u_{i+1}$ is the base of a single pipe, which is outside of $\Delta_u$, as the other base would have to be $u_iu_{i+1}$.
		Due to the previous claim and since $U$ is a cycle, there are two consecutive clusters, $u_j$ and $u_{j+1}$, in $U$ such that $u_j$ was created by an earlier invocation of \textsf{clusterExpansion} than $u_{j+1}$. By \ref{b3}, $u_{j+1}$ is a base for $u_ju_{j+1}$. Indeed, $u_{i+1}u_i$ is the unique pipe  at $u_{i+1}$ outside of $\Delta_u$ for which $u_{i+1}$ is also a base. This contradicts our assumption, and proves that Case~2 does not occur.
	\end{proof}
	
	\begin{lemma}
		\label{lem:all-safe}
		The following hold for the instance $\phi':G'\rightarrow H'$ at the end of Phase~2:
		\begin{enumerate}
			\item\label{11} every pipe in $E(H')$ is empty or useless,
			\item\label{22} every component of $G'$ is a generalized cycle, and
			\item\label{33} any two components of $G'$ are mapped to the same or two vertex-disjoint cycles in $H'$.
		\end{enumerate}
	\end{lemma}
	\begin{proof}
		{\bf \ref{11}}.
		When the while loop of Phase~2 terminates, there are no useful pipes in $H'$. Consequently, every safe pipe is useless.
		By Lemma~\ref{lem:emptysafe}, every unsafe pipe is empty at that time. Overall, every pipe is empty or useless, as claimed.
		
		{\bf \ref{22}}. Since every pipe $uv\in E(H')$ is empty or useless, every component in every cluster is stable.
		It follows that every component of $G'$ must be a generalized cycle.
		
		{\bf \ref{33}}. Since every pipe $uv\in E(H')$ is empty or useless, every cluster is incident to at most two nonempty pipes.
		Consequently, any two generalized cycles of $G'$ are mapped to either the same cycle or two disjoint cycles in $H'$.
	\end{proof}
	
	\begin{lemma}
		\label{lem:correctness}
		The algorithm reports whether $\phi$ is a weak embedding.
	\end{lemma}
	\begin{proof}
		By Lemmas~\ref{lem:clusterExpansion} and \ref{lem:pipeExpansion}, every operation either reports that the instance is not a weak embedding and halts, or produces an instance equivalent to the input $\phi$. Consequently, if any operation finds a negative instance, then $\varphi$ is not a weak embedding. Otherwise the while loop in Phase~2 terminates, and yields an instance $\varphi':G'\rightarrow H'$ equivalent to $\varphi$.
		By Lemma~\ref{lem:all-safe}, every component of $G'$ is a generalized cycle, any two of which are mapped to the same or vertex-disjoint cycles in $H'$.
		
		\begin{figure}[htbp]
			\centering
			\def\svgwidth{.3\textwidth}
			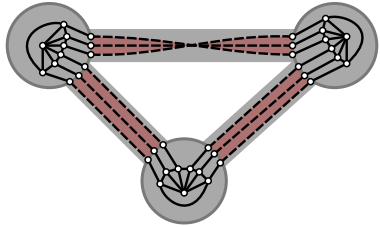
			\caption{\small
				A nonplanar generalized cycle that does not embed in an annulus.
				It embeds in the M\"obius band or the projective plane.}
			\label{fig:thickCycle}
		\end{figure}
		
		The strip system of the subgraph $\varphi'(C)$ of $H'$, where $C$ is a generalized cycle, is homeomorphic to the annulus, since $M$ is orientable.
		If $C$ is nonplanar, then $\varphi'$ is clearly not a weak embedding (see Fig.~\ref{fig:thickCycle}).
		Furthermore, a generalized cycle $C$ that winds around $\phi'(C)$ several times cannot be embedded in that annulus (Fig.~\ref{fig:intro}(d)). However, one or more generalized cycles that each wind around $\phi'(C)$ once can be embedded in nested annuli in the strip system $\mathcal{H}'$ (Fig.~\ref{fig:intro}(c)). This completes the proof of correctness of Algorithm$(\phi)$.
	\end{proof}

	\subsection{Efficient Implementation}
	\label{ssec:alg2}
	
	Recall that \textsf{pipeExpansion}$(uv)$ first contracts the pipe $uv$ into the new cluster $\cluster{uv}\in V(H^*)$. Each stable component $C$ of $G_{\cluster{uv}}$ is composed of two stable components, in $G_u$ and $G_v$, respectively. Then \textsf{clusterExpansion}$(uv)$ performed at the end of  \textsf{pipeExpansion}$(uv)$ splits $C$ into two stable components in two new clusters. That is, two adjacent stable components in $G_u$ and $G_v$ are replaced by two adjacent stable components in two new clusters. Consequently, we cannot afford to spend $O(|E_{\cluster{uv}}|)$ time for \textsf{pipeExpansion}$(uv)$.
	We introduce auxiliary data structures to handle stable components efficiently: we use \emph{set operations} to maintain a largest set of stable components in $O(1)$ time. A dynamic variant of the heavy path decomposition yields an $O(m\log m)$ bound on the total time spent on stable components in the main loop of the algorithm.
	
	\paragraph{Data structures for stable components.}
	For each pair $(u,uv)$ of a cluster $u\in V(H)$ and an incident pipe $uv\in E(H)$,
	we store a set $L(u,uv)$ of all stable components of $G_u$ adjacent to a (thick) edge in $\phi^{-1}[uv]$.
	For every cluster $u\in V(H)$, let $w^*(u)$ be a neighbor of $u$ maximizing the size of $L(u,uw^*(u))$;
	we maintain a pointer from $u$ to the set $L(u,uw^*(u))$.
	The total number of sets $L(u,uv)$ and the sum of their sizes are $O(m)$.
	We can initialize them in $O(m)$ time.
	For each cluster $u\in V(H)$, we store the components of $G_u$ in two sets:
	a set of stable and unstable components, respectively, each stored as the ID of a representative edge of the component.
	Hence, for each pipe we can determine in $O(1)$ time whether $uv$ is useful or useless.
	Every stable component in $G_u$ has a pointer to the two sets $L(u,.)$ in which it appears.
	
	\begin{figure}[h]
		\centering
		\def\svgwidth{.65\textwidth}
		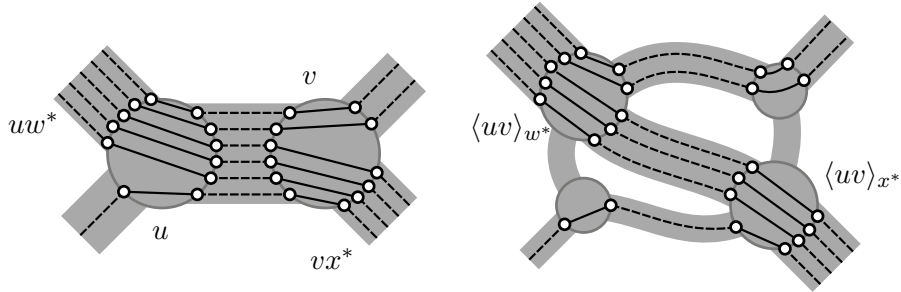
		\caption{\small
			Implementation of $\textsf{pipeExpansion}(uv)$ for stable components.
		}
		\label{fig:implementation}
	\end{figure}
	
	\smallskip
	In each iteration of Phase~2, we implement  \textsf{pipeExpansion}$(uv)$ for a useful pipe $uv$ as follows.
	Let $G^\diamond_{\cluster{uv}}$ be the graph formed by the unstable components of $G_{\cluster{uv}}$ and stable components 
	$C$ where $\sigma_C(uv)>1$ (such as paths that ``zig-zag'' between $u$ and $v$).
	Compute $G^\diamond_{\cluster{uv}}$ using DFS starting from each unstable component of
	$G_u$ and $G_v$.
	If a stable component of $G_u$ and $G_v$ is absorbed by a new unstable component of $G_{\cluster{uv}}$,
	delete it from the set $L(u,.)$ or $L(v,.)$ in which it appears.\\
	\textbf{(a)} We perform \textsf{simplify} and \textsf{clusterExpansion} only on $G^\diamond_{\cluster{uv}}$.\\
	\textbf{(b)} We handle the remaining stable components of $G_{\cluster{uv}}$ as follows (refer to Fig.~\ref{fig:implementation}).
	Since $uv$ is safe, every stable component of $G_u$ and $G_v$ appears in $L(u,uv)$ and $L(v,uv)$, respectively.
	For all $w$, $w\notin \{v,w^*\}$, where $w^*=w^*(u)$, move all components from $L(u,uw)$ to the cluster $\cluster{uv}_w$.
	Similarly, for all $x$, $x\notin\{u,x^*\}$, where $x^*=w^*(v)$, move all components from $L(v,vx)$ to the cluster $\cluster{uv}_x$.
	Let $L'(u,uv)$ and $L'(v,uv)$ be the sets obtained from $L(u,uv)$ and $L(v,uv)$ after this step.
	Notice that the remaining stable components of $G_u$ and $G_v$ in $L'(u,uv)$ and $L'(v,uv)$ appear in $L(u,uw^*)$ and $L(v,vx^*)$, respectively. Move them to the clusters $\cluster{uv}_{w^*}$ and $\cluster{uv}_{x^*}$ by renaming all data structures of $u$ and $v$, respectively,
	which can be done in $O(1)$ time.
	Update $H$ so that it reflects the changes in $G$,
	by creating appropriate pipes between the clusters in $\partial\Delta_{\cluster{uv}}$, checking for crossings by local planarity testing inside $\Delta_{\cluster{uv}}$ (similar to Step~4 of \textsf{clusterExpansion}). \\
	\textbf{(c)} We add the new stable components (created by \textsf{clusterExpansion} from unstable components) that fit the definition of $L(\cluster{uv}_{w^*},\cluster{uv}_{w^*}\cluster{uv}_{x^*})$ to the current set $L'(u,uv)$
	to obtain the set $L(\cluster{uv}_{w^*},\cluster{uv}_{w^*}\cluster{uv}_{x^*})$.
	Analogously, add new stable components to the current sets $L'(v,uv)$, $L(u,uw)$ where $w\neq v$, and $L(v,vx)$ where $x\neq u$ to obtain the new sets $L(\cluster{uv}_{x^*},\cluster{uv}_{x^*}\cluster{uv}_{w^*})$, $L(\cluster{uv}_w,\cluster{uv}_w w)$, and $L(\cluster{uv}_x,\cluster{uv}_x x)$, respectively.
	Compute any other sets of stable components from scratch.
	
	We are now ready to show that the running time of the algorithm improves to $O(m\log m)$.
	
	\begin{lemma}\label{lem:time}
		Our implementation of the algorithm runs in $O(m\log m)$ time.
	\end{lemma}
	\begin{proof}
		Phases~1 and 3 take $O(m)$ time by Lemmas~\ref{lem:termination}, \ref{lem:clusterExpansion}, and \ref{lem:all-safe}. The while loop in Phase~2 terminates after $O(m)$
		iterations by Corollary~\ref{cor:termination}. Using just Lemmas~\ref{lem:normalize} and \ref{lem:clusterExpansion}, each iteration of Phase~2 would take $O(m)$ time leading to an overall running time of $O(m^2)$. We define a new potential function for an instance $\phi:G\rightarrow H$ to show that each iteration of Phase~2 takes $O(\log m)$ amortized time.
		
		For every $u\in V(H)$, let $\mathcal{L}(u)$ be the number of stable components in $G_u$.
		Let $s$ be the number of stable components created from the beginning of Phase~2 up to the current iteration.
		We define a new potential function as\footnote{All logarithms are of base 2.}
		
		$$\Phi_2(\phi)=\Phi_1(\phi)+ (12m-s)\log(48m)+ \sum_{u\in V(H)}\mathcal{L}(u)\log \mathcal{L}(u).$$
		By Corollary~\ref{cor:termination} the second term is nonnegative.
		Note that $\Phi_2(\phi)=O(m\log m)$ since $\Phi_1(\phi)=O(m)$ and $\sum_{u\in V(H)}\mathcal{L}(u) = O(m)$.
		We show that $\Phi_2$ strictly decreases in Phase~2. As argued above (cf. Lemma~\ref{lem:termination}), $\Phi_1$ strictly decreases. The second term of $\Phi_2$ can only decrease since $s$ increments when stable components are created (but never decrements).
		The term $\sum_{u\in V(H)}\mathcal{L}(u)\log \mathcal{L}(u)$ increases when new stable components are created. However, this increase is offset by the decrease in the second term of $\Phi_2$. It suffices to consider the case that $\mathcal{L}(u)$ increments from $k$ to $k+1$. Then
		\begin{align*}
			k\log k + \log (48m)
			& \geq k\log k + \log (4s)\\
			&\geq \left(k\log (k+1) - k \log \frac{k+1}{k}\right) + (2+\log (k+1))\\
			&\geq (k+1)\log (k+1) + 2 - \log \left(1+\frac{1}{k}\right)^k\\
			&\geq (k+1)\log (k+1) + 2 - \log e\\
			&> (k+1)\log (k+1),
		\end{align*}
		that is, the decrease of $\log(48m)$ in the second term offsets the increase of $(k+1)\log(k+1)-k\log k$ in the third term.
		
		We next show that the time spent on each iteration of the while loop in Phase~2 is bounded above by a constant times the  decrease of the potential $\Phi_2$. This will complete the proof since the potential $\Phi_2$ is nonincreasing throughout the execution of the algorithm as we have just shown.
		
		For a useful pipe $uv\in E(H)$, let $\Lambda(uv)=\Lambda$ be the collection of sets consisting of all $L(u,uw)$ where $w\notin\{v,w^*(u)\}$ and all $L(v,vx)$ where $x\not\in\{u,w^*(v)\}$.
		Our implementation of \textsf{pipeExpansion}$(uv)$ spends $O(|E(G^\diamond_{\cluster{uv}})|+\deg(\cluster{uv}))$ time to process the components of $G^\diamond_{\cluster{uv}}$, by Lemma~\ref{lem:clusterExpansion}, and $O(1+\sum_{L\in \Lambda}|L|)$ time to process the remaining stable components.
		However, by~\ref{b4}, $\deg(\cluster{uv})=O(|E(G^\diamond_{\cluster{uv}})|+1+\sum_{L\in \Lambda}|L|)$.
		Hence the running time of \textsf{pipeExpansion}$(uv)$ is
		
		\begin{equation}\label{eq:time}
			O\left(|E(G^\diamond_{\cluster{uv}})|+1+\sum_{L\in \Lambda}|L|\right).
		\end{equation}

		First, let $C$ be a component of $G^\diamond_{\cluster{uv}}$. We have seen (in the proof of Lemma~\ref{lem:termination}) that \textsf{pipeExpansion}$(uv)$  decreases $N_\sigma$ by at least $\sigma_C(uv)-1$ due to edges in $C$.
		By the definition of safe pipes, $C$ contains a nonempty set $A$ of components of $G_u$ and a nonempty set $B$ of components of $G_v$.
		By Lemma~\ref{obs:1}, each component $C'\in A\cup B$ contains $O(t_{C'})$ edges, where $t_{C'}$ is the number of terminals of $C'$.
		Notice that $\sigma_C(uv)\ge\sum_{C'\in A} (t_{C'}-3)$, because at most three out of $t_{C'}$ terminals are not adjacent to a pipe-edge in $\phi^{-1}[uv]$ by the definition of safe pipes (equality occurs when a thick edge is incident to $C'$ in a pipe other than $uv$).
		As such, for a component $C$ in $G^\diamond_{\cluster{uv}}$, the contribution of $C$ to $N_\sigma-Q$ decreases by $\Omega(|E(C)|)$.
		Therefore, the first term of (\ref{eq:time}) is charged to the decrease in $N_\sigma-Q$.
		By Lemma~\ref{lem:termination}, summation over pipe expansions in Phase~2 yields $O(m)$.
		This takes care of steps \textbf{(a)} and~\textbf{(c)} of the efficient implementation (Section~\ref{ssec:alg2}). It remains to bound the time complexity of step~\textbf{(b)}, which deals exclusively with stable components.
		
		Second, we show that the time that \textsf{pipeExpansion}$(uv)$ spends on stable components
		is absorbed by the decrease in the last term of $\Phi_2$.
		When we move the components of $L(u,uw)$, $w\notin\{v,w^*\}$, to the cluster $\cluster{uv}_w$, we spend linear time on all but a maximal set, which can be moved in $O(1)$ time using a set operation. In what follows we show that a constant times the corresponding decrease in the term  $\sum_{u\in V(H)}\mathcal{L}(u)\log \mathcal{L}(u)$ subsumes this work.
		
		We adapt the analysis from the classic heavy path decomposition. Suppose we partition a set of size
		$k=\mathcal{L}(u)=\mathcal{L}(v)$ into $\ell$ subsets of sizes $k_1\geq \ldots \geq k_\ell$.
		Note that $k_j\le k/2$ for $j\ge 2$. Then
		
		$$ k\log k
		= \sum_{i=1}^\ell k_i\log k\\
		\geq k_1\log k_1 + \sum_{j=2}^\ell k_j\log (2k_j)\\
		= \sum_{i=1}^\ell k_i\log k_i + \sum_{j=2}^\ell k_j.
		$$
		Hence, the decrease in $\sum_{u\in V(H)}\mathcal{L}(u)\log \mathcal{L}(u)$, which is equal to $k\log k-\sum_{i=1}^\ell k_i\log k_i$, is bounded from below by $k-k_1$. Therefore if we spend $O(1)$ time on a maximal subset (of size $k_1$), we can afford to spend linear time on all other subsets.
		Thus, the decrease in $\sum_{u\in V(H)}\mathcal{L}(u)\log \mathcal{L}(u)$ subsumes the actual work and this concludes the proof.
	\end{proof}
	
	This completes the proof of Theorem~\ref{thm:main}(i). Part~(ii) of Theorem~\ref{thm:main} is shown in Section~\ref{sec:reverse}.

	\section{Constructing an embedding}
	\label{sec:reverse}
	
	Our recognition algorithm in Section~\ref{sec:alg} decides in $O(m\log m)$ time whether a given instance $\varphi$ is a weak embedding. However, if $\varphi$ turns out to be a weak embedding, it does not provide an embedding $\psi_{\varphi}$,
	since at the end of the algorithm we have an equivalent ``reduced'' instance $\varphi'$ at hand. In this section, we show how to compute the combinatorial representation of an embedding $\psi_\varphi$ for the input $\varphi$.
	
	Assume that $\varphi:G\rightarrow H$ is a weak embedding. By Lemmas~\ref{lem:correctness} and \ref{lem:time}, we can obtain a combinatorial representation of an embedding $\pi_{\phi'}\in \Pi(\phi')$ of the instance $\varphi':G'\rightarrow H'$ produced by the algorithm at the end of Phase~2 in $O(n\log n)$ time.
	
	We sequentially reverse the steps of the algorithm, and maintain combinatorial embeddings for all intermediate instances until we obtain a combinatorial representation $\pi\in \Pi(\phi)$. By Lemma~\ref{lem:comb-reconstruction}, we can then obtain an embedding $\psi_{\phi}:G\rightarrow\mathcal{H}$ in $O(m)$ time. Reversing a \textsf{clusterExpansion}$(u)$ operation is trivial: the total orders of pipes within $\Delta_u$ can be ignored and the total order for every pipe $uv$ is the same as the order for $u_vv$. This can be done in $O(\deg(u))$ time.
	
	Let $\varphi^{(1)}:G^{(1)}\rightarrow H^{(1)}$ be the input instance of \textsf{pipeExpansion}$(uv)$, $\varphi^{(2)}:G^{(1)}\rightarrow H^{(2)}$ be the instance obtained by contracting $uv$ and  $\varphi^{(3)}:G^{(3)}\rightarrow H^{(3)}$ be the instance after \textsf{clusterExpansion}$(\cluster{uv})$. By the previous argument, the total orders of pipe-edges $\pi^{(2)}(\cluster{uv}w)$ of all pipes $\cluster{uv}w$ can be obtained from a combinatorial representation $\pi^{(3)}$ in $O(\deg(\cluster{uv}))$ time. These orders also correspond to $\pi^{(1)}(uw)$ and $\pi^{(1)}(vx)$ for pipes $uw$ and $vx$ in $\phi^{(1)}$ where $w\neq v$ and $x\neq u$.
	
	To obtain an order $\pi^{(1)}(uv)$, we embed $G_{\cluster{uv}}^{(1)}$ into $D_{\cluster{uv}}$ using Lemma~\ref{lem:comb-reconstruction} in $O(|E(G_{\cluster{uv}}^{(1)})|)$ time and find the Jordan curve defined in the proof of Lemma~\ref{lem:pipeExpansion}. The order $\pi^{(1)}(uv)$ of the pipe-edges in $\varphi^{-1}[uv]$ is given by the order in which the Jordan curve intersects these edges.
	This takes $O(|E(G_{\cluster{uv}}^{(1)})|)$ time by Theorem~\ref{thm:belyi}, though we first need to obtain an embedding of $G^{(1)}$, where triangles of wheels in $G_{\cluster{uv}}^{(1)}$ and 4-cycles induced by thick edges in and incident to $G_{\cluster{uv}}^{(1)}$ are empty.
	This can be done by changing the rotation at the vertices of the wheels and 4-cycles corresponding to thick edges one by one
	in $O(|E(G_{\cluster{uv}}^{(1)})|)$ time, since $uv$ is safe in the resulting instance.
	
	However, this would lead to a $O(m^2)$ worst case time complexity because of stable components.
	We show how to reduce the running time to $O(m\log m)$.
	Let us call a pipe-edge (thick edge) \textbf{stable}
	if it connects two stable components in two adjacent clusters.
	In each total order $\pi^*(uv)$ for a pipe $uv$ in an instance $\phi^*$,
	we arrange maximal blocks of consecutive stable edges into a \textbf{bundle}
	that takes a single position in the order, and we store the order among the stable edges
	in the bundle in a separate linked list.
	We can substitute each bundle with one \textbf{representative} stable component.
	Then, using $\pi^{(2)}(\cluster{uv}w)$ for all pipes $\cluster{uv}w$ incident to $\cluster{uv}$ we can obtain a list of at most $\deg(\cluster{uv})+c$ representative stable components, where $c$ is the number of components of $G^\diamond_{\cluster{uv}}$.
	We can proceed by embedding all components in $G^\diamond_{\cluster{uv}}$ and the representatives of the remaining stable components in $D_{\cluster{uv}}$.
	Obtaining a Jordan curve that encloses all vertices in $(\phi^{(1)})^{-1}[u]$ now takes $O(E(G^\diamond_{\cluster{uv}})+\deg(\cluster{uv}))$ time.
	The order in which the Jordan curve crosses the edges in $G_{\cluster{uv}}$ defines $\pi^{(1)}(uv)$, where the size of $\pi^{(1)}(uv)$ is $O(E(G^\diamond_{\cluster{uv}})+\deg(\cluster{uv}))$ and each pipe-edge obtained from a representative stable component represents a bundle of stable edges. We can merge consecutive bundles of stable edges in $\pi^{(1)}(uv)$, as needed, in $O(E(G^\diamond_{\cluster{uv}})+\deg(\cluster{uv}))$ time.
	By \ref{b4}, this running time is bounded above by the running time of our implementation of \textsf{pipeExpansion}$(uv)$, as argued in the proof of Lemma~\ref{lem:time}.
	Therefore, we can reverse every operation and obtain an embedding
	$\psi_{\phi}:G\rightarrow\mathcal{H}$ in $O(m\log m)$ time.
	This completes the proof of Theorem~\ref{thm:main}(ii).

	\section{Algorithm for nonorientable surfaces}
	\label{sec:nonorientable}
	
	We show that our algorithm can be adapted to recognize weak embeddings $\varphi:G\rightarrow H$ when $H$ is embedded in a \emph{nonorientable} manifold $M$. We discuss the adaptation to nonorientable surfaces in a separate section to reduce notational clutter in Sections~\ref{sec:pre}--\ref{sec:reverse}.
	
	First, we adapt the definition of the strip system. The embedding of a graph $H$ into a (orientable or nonorientable) surface $M$ is given by a rotation system that specifies, for each vertex of $H$, the ccw cyclic order of incident edges, and a signature
	$\lambda: E(H)\rightarrow\{-1,1\}$. To define the strip system $\mathcal{H}$, we proceed exactly as in the case that $M$ is orientable (Section~\ref{sec:pre}), except that for every edge $e=uv$, if $\lambda(e)=-1$, we identify $A_{u,v}$ with $\partial R_{uv}$ via an orientation reversing homeomorphism and $A_{v,u}$ with $\partial R_{uv}$ via an orientation preserving homeomorphism, or vice-versa.
	If we represent $M$ as a sphere with a finite number holes, that are turned into \textbf{cross-caps},
	the signature of an edge is interpreted as the parity of the number of times an edge passes through a cross-cap.
	
	Second, we adapt the operation of cluster expansion as follows. We put $\lambda(u_vv):=\lambda(uv)$ for all neighbors $v$ of $u$, and  $\lambda(e):=1$ for all newly created edges $e$.
	
	Third, we adapt the operation \textsf{pipeExpansion}$(uv)$ as follows. If $\lambda(uv)=-1$,
	before creating the cluster $\cluster{uv}$, we flip the value of $\lambda$ from $-1$ to 1, and vice-versa,
	for every  edge of $H$ adjacent to $u$.
	This corresponds to pushing the edge $uv$ off all the cross-caps that it passes through.
	Then the values of $\lambda$ on the edges incident to $\cluster{uv}$ in $H^*$ are naturally inherited from the values of $\lambda$ on the edges adjacent to $u$ and $v$ in $H$. The value of $\lambda$ for all other edges in $H^*$ remain the same as in $H$.
	
	The first two phases of the algorithm remain the same except that they use the adapted operations of cluster and pipe expansion.
	Phase~3 is modified as follows.
	If a generalized cycle $C$ in $G'$ satisfies $\Pi_{e\in E(\varphi'(C))}\lambda(e)=1$ (that is, the strip system of $\varphi'(C)$ is homeomorphic to an annulus), we proceed as in the orientable case.
	The modifications affect only the generalized cycles $C$ such that $\Pi_{e\in E(\varphi'(C))}\lambda(e)=-1$, or in other words generalized cycles $C$, for which the strip system of $\varphi'(C)$ is homeomorphic to the M\"obius band.
	For such a generalized cycle $C$, we report that the instance is negative if $C$ winds more than two times around $\varphi'(C)$; or if all pipe edges of $C$ are thick edges, the underlying graph of $C$ is planar (resp., not planar), and $C$ winds exactly once (resp., twice) around $\varphi'(C)$.
	Furthermore, we report that the instance is negative if there exist two distinct generalized cycles $C_1$ and $C_2$ in $G'$ winding once around $\varphi'(C)$ such that $\varphi'(C)=\varphi'(C_1)=\varphi'(C_2)$ in $H'$.
	Else we can report that $\varphi$ is a weak embedding at the end.

	The correctness of the algorithm in the M\"obius band case is implied by a stronger statement in \cite[Section 8]{FK18_ht}, but we can easily verify it by the following contradictions.
	Suppose two or more cycles wind exactly once around $\varphi'(C)$.
	In any embedding $\varphi$, the total orders in $\pi$
	(defined in Section~\ref{sec:combRep}) reverses in each traversal of the cycle, which is a contradiction.
	Similarly, if $C$ winds $k>2$ many times around $\varphi'(C)$ we  divide it into $k$ paths $P_1,\ldots, P_k$ sharing end vertices each of which is mapped injectively into $\varphi'(C)$. We assume that in $\pi$ the paths $P_1,\ldots, P_k$ appear in the given order up to the choice of orientation. The order of $P_i$'s in $\pi$ reverses with respect to a fixed orientation if we traverse  $\varphi'(C)$. Hence, $P_1$ must precede and follow $P_k$ along $C$ and therefore $k\le 2$, a contradiction.

	\section{Conclusions}
	\label{sec:con}
	
	We have shown (Theorem~\ref{thm:main}) that it takes $O(m\log m)$ time to decide whether a piecewise linear simplicial map $\varphi:G\rightarrow H$ from an abstract graph $G$ with $m$ edges to an embedded graph $H$ is a weak embedding (i.e., whether, for every $\eps>0$, there exists an embedding $\psi_{\eps}:G\rightarrow \mathcal{H}$ of $G$ into a neighborhood $\mathcal{H}$ of $H$ such that $\|\varphi-\psi_\eps\|<\eps$).
	The only previously known algorithm for this problem takes $O(n^{2\omega})\leq O(n^{4.75})$ time, where $\omega$ is the matrix multiplication constant~\cite{FK18_ht}, and until recently no polynomial-time algorithm was available even in the special case that $H$ is embedded in the plane. Only the trivial lower bound of $\Omega(m)$ is known for the time complexity of recognizing weak embeddings in our setting. Closing the gap between $\Omega(m)$ and $O(m\log m)$ remains open.
	
	If $\varphi:G\rightarrow H$ is a continuous map, but not necessarily simplicial (i.e., an image of an edge may pass through vertices of $H$), then the running time increases to $O(mn\log (mn))$, where $n=|V(G)|$ (cf.~Corollary~\ref{cor:nonsimplicial}).
	In the special case that $G$ is a cycle and $H$ is a planar straight-line graph, an $O(m\log m)$-time algorithm was recently given in~\cite{AAET17}. It remains an open problem whether a similar improvement is possible for arbitrary $G$ and $H$.

	\paragraph{Atomic Embeddings.}
	An interesting generalization of our problem is deciding whether a given 2-dimensional simplicial complex embeds into some 3-dimensional manifold, also known as the \emph{thickenability of 2-dimensional simplicial complexes}.
	Indeed, \cite[Lemma]{Skop94_thick} implies\footnote{The Lemma in~\cite{Skop94_thick} is stated only for connected graphs and only in the case that the target surface is a sphere. However, it is easy to see that an analogous statement holds for disconnected graphs and for orientable surfaces of arbitrary genus.} that a polynomial-time algorithm for the thickenability of 2-dimensional simplicial complexes would directly translate into a polynomial-time algorithm for our problem. Studying this more general problem is one of the next natural steps in our investigation: Currently, we do not know whether this problem is tractable. Furthermore, due to the same argument, a polynomial-time algorithm for deciding whether a 2-dimensional simplicial complex embeds in $\mathbb{R}^3$ would already imply a polynomial-time algorithm for our problem if we restrict ourselves to orientable surfaces. However, this problem is NP-hard~\cite{dMRST18_embedR3}, so the existence of such an algorithm is highly unlikely.
	
	The thickenability problem, originally due to Neuwirth~\cite{Skop94_thick}, can be seen as the following variant of the problem of deciding whether a simplicial map $\varphi:G\rightarrow H$ is a weak embedding\footnote{In fact, for Neuwirth's algorithm it is enough to assume that $G$ is a finite union of pairwise disjoint 3-cycles.}, which we call  the \textbf{atomic embeddability problem}: We are given an abstract graph $H$, and a strip system $\mathcal{H}$ as a 2-dimensional surface without boundary, partitioned into regions representing ``clusters'' and ``pipes.''
	Every vertex $u\in V(H)$ is represented by a surface $S_u$ with boundary (instead of a disk $D_u$), obtained from a 2-sphere by removing disjoint open disks (holes), where the number of holes equals $\deg(u)$ and each hole corresponds to an edge in $H$ incident to $u$.
	Every edge $uv\in E(H)$ is associated with a cylinder $C_{uv}$ (rather than a rectangle $R_{uv}$) whose boundary components are homeomorphically identified with the boundaries of the corresponding holes on $S_u$ and $S_v$, respectively.
	%
	%
	The technique that we developed in Sections~\ref{sec:pre}--\ref{sec:alg} is not directly applicable to the setting of atomic embeddings. It remains an open problem to recognize atomic embeddability efficiently.
	
%
	
	\bibliographystyle{plain}
	\bibliography{wembedding}
	
\end{document}